\documentclass[12pt,english]{article}

\usepackage[T1]{fontenc}
\usepackage[margin=1.25in]{geometry}
\usepackage[latin9]{inputenc}
\usepackage{color}
\usepackage[dvipsnames]{xcolor}
\usepackage{enumerate}
\usepackage{enumitem} \usepackage{hyperref}
\usepackage{booktabs} 

\definecolor{darkblue}{rgb}{0, 0, 0.5}
\hypersetup{
     colorlinks = true,
     citecolor = Maroon,
     urlcolor = darkblue,
     linkcolor = darkblue
}
  
\usepackage{amsmath}
\usepackage{amsthm}
\usepackage{amssymb}
\usepackage{cases}
\usepackage{bm}
\usepackage[authoryear,longnamesfirst]{natbib}

\usepackage{setspace}
\usepackage{graphicx}
\usepackage{float}
\usepackage[flushleft]{threeparttable}
\usepackage{pdflscape}
\usepackage{multicol}

\makeatletter
\theoremstyle{plain}
  \newtheorem{rem}{\protect\remarkname}

  \theoremstyle{plain}
  \newtheorem{assumption}{\protect\assumptionname}
  
\theoremstyle{plain}
\newtheorem{thm}{\protect\theoremname}
\newtheorem*{thm*}{\protect\theoremname}

  \theoremstyle{plain}
  \newtheorem{lem}{\protect\lemmaname}[section]
    
  \theoremstyle{remark}

\makeatother

\usepackage{babel}
  \providecommand{\assumptionname}{Assumption}
  \providecommand{\claimname}{Claim}
  \providecommand{\lemmaname}{Lemma}
  \providecommand{\remarkname}{Remark}
\providecommand{\theoremname}{Theorem}
  \providecommand{\examplename}{Example}
\providecommand{\corname}{Corollary}

\usepackage{graphicx}
\usepackage{subcaption}
\graphicspath{{figures/}}                  \DeclareGraphicsExtensions{.pdf,.png,.eps}         \title{Refined Cluster Robust Inference\thanks{We would like to thank A. Colin Cameron, Antoine A. Djogbenou, and participants of UC Davis Econometrics seminar and SETA2026 for helpful comments. Ura acknowledges funding from A. Colin Cameron Associate Professor Research Award.}}
\author{\setcounter{footnote}{1}
Bulat Gafarov\thanks{Department of Agricultural and Resource Economics, University of California, Davis. Email: \href{bgafarov@ucdavis.edu}{bgafarov@ucdavis.edu}}
\and 
Takuya Ura\thanks{Department of Economics, University of California, Davis. Email: \href{takura@ucdavis.edu}{takura@ucdavis.edu}}
}
\begin{document}
\maketitle
\begin{abstract}
\linespread{1.2}

It has become standard for empirical studies to conduct inference robust to cluster dependence and heterogeneity.  With a small number of clusters, the normal approximation for the $t$-statistics of regression coefficients may be poor. This paper tackles this problem using a critical value based on the conditional Cram\'er-Edgeworth expansion for the $t$-statistics. The proposed critical value guarantees third-order refinement, and it does not require resampling because it is a closed-form function of the estimated score skewness and kurtosis. Simulations show that our proposal can make a difference in size control with as few as 10 clusters.

\begin{description}
\item[Keywords:] Cluster robust inference, Cram\'er-Edgeworth expansion, Asymptotic refinement
\item[JEL Classification:]  C12, C21
\end{description}
\end{abstract}
\newpage
 \section{Introduction}

Cluster robust inference has become a standard practice in applied microeconometrics. Robustness to arbitrary dependence within a cluster often comes at the cost of a reduction in the effective sample size. As a result, despite a large overall sample size, the researchers have to account for the non-Gaussian distribution of the $t$-statistic for significance tests, as was the case in classical statistics \citep{student1908probable}. In particular, it is well-known that the standard normality-based inference method may lead to over-rejection when the   number of clusters  is small \citep{cameron2015practitioner,cameronMiller2025inference,mackinnon2023cluster}.

In this paper, we propose a new analytical correction to critical values based on inverting the Cram\'er-Edgeworth expansion \citep{Edgeworth1905,Cramer1928} of the $t$-statistic null distribution (i.e., using the expansion of \cite{cornish1938moments}). We  first derive the Cram\'er-Edgeworth term up to the second order and then adjust  the critical value based on the estimated Cram\'er-Edgeworth term.
In other words, we develop the approach of \cite{hall1983inverting} in the setup of linear regression with heterogeneous clusters.\footnote{\cite{hall1983inverting} considers independent and identically distributed data. Moreover, in this paper, we consider a two-sided alternative hypothesis and the second-order Cram\'er-Edgeworth expansion. As a result, we do not need to apply the Cram\'er-Edgeworth expansion of the estimated coefficient recursively, as in \cite{hall1983inverting}.} As special sub-cases, our approach also nests inference on the sample mean of non-identically distributed data and regression coefficients in i.i.d. cross-sectional regression. The resulting inference achieves third-order asymptotic refinement in the sense that the actual test size differs from the prespecified significance level $\alpha$ only by $o(G^{-1})$ with $G$ clusters.
Our simulation studies support this theoretical size accuracy. In our designs with $G\geq 10$, the proposed critical value has better size accuracy than the normal critical value. To the best of our knowledge, this is the first paper to use the estimated Cornish-Fisher expansion for inference with observations from non-identical distributions.

The benefits of the proposed method can be better understood when comparing it with the existing methods. Among them, the pairs cluster bootstrap might be a natural choice for small sample inference, but our proposed method has a few advantages.\footnote{The residual bootstrap cannot be applied for cluster robust inference when cluster sizes vary across clusters.} First and most importantly, existing simulation studies show that the pairs cluster bootstrap does not perform well with a small number of clusters. For example, \cite{cameron2008bootstrap} use a simulation design based on \cite{bertrand2004much} and explain that the poor performance of the pairs cluster bootstrap is due to the fact that the resampled values of the Gram matrix $\frac{1}{G}\sum_{g=1}^GX_{g}'X_{g}$ are nearly singular.\footnote{\citet[p.394]{djogbenou2019asymptotic} also state that the ``fundamental problem'' of the pairs cluster bootstrap is that the resampled data has different values of $X_1,\ldots,X_G$ than the original data.} Our approach of analytically inverting the Cram\'er-Edgeworth expansion avoids this problem by not resampling $\frac{1}{G}\sum_{g=1}^GX_{g}'X_{g}$ while achieving asymptotic refinement. Second, the pairs cluster bootstrap uses independence across $X_1,\ldots,X_G$, while our expansion does not require it. We allow $X_1,\ldots,X_G$ to be correlated, e.g., through adaptive randomization. Third, the standard proof for the pairs cluster bootstrap's asymptotic refinement \citep[e.g.,][]{liu1988bootstrap,hall1992bootstrap} excludes discrete regressors. To apply the results from \citet[Ch.5]{hall1992bootstrap} to the regression framework, we need to impose the Cram\'er condition on the regressors; however, the Cram\'er condition fails for discrete random variables \citep[p.207]{bhattacharya2010normal}.\footnote{We could not find a sufficient condition for the pairs cluster bootstrap's asymptotic refinement that allows for discrete regressors and non-identical distributions.
A weaker version of the Cram\'er condition has been proposed, e.g., in \cite{bai1991edgeworth}, but the pairs cluster bootstrap's asymptotic refinement without the classical Cram\'er condition is beyond the scope of this paper.}  Lastly, our proposed critical value has a closed-form expression, and for this reason, its computation is much faster than resampling. 

A number of other methods have been proposed to address the limitations of the pairs cluster bootstrap. For example, \cite{cameron2008bootstrap} propose the wild cluster bootstrap and demonstrate its good finite-sample properties in simulations. \cite{djogbenou2019asymptotic} provide the asymptotic size control of the wild cluster bootstrap, and \cite{canay2021wild} show the size control of the wild cluster bootstrap even with a fixed number of clusters as long as there are a large number of observations per cluster. However, Theorem 5.2 of \cite{djogbenou2019asymptotic} shows that the wild cluster bootstrap does not achieve similar asymptotic refinement when the score has non-zero skewness.\footnote{In this paper, we do not use the Cram\'er-Edgeworth expansion by  \cite{djogbenou2019asymptotic} directly. Remark \ref{rem:different_edgeworth} explains more details.} Our simulation results in Section \ref{sec:skewed_MONTECARLO_DGP} confirm this result of \cite{djogbenou2019asymptotic}, and our proposed method demonstrates better size control than the wild cluster bootstrap in such cases.

There are  papers that provide closed-form critical-value corrections by assuming normal errors  \citep{bell2002bias,bester2011inference,imbens2016robust,young2016improved,hansen2021exact,hansen2026jackknife}. Similarly to these, the critical value proposed in this paper has a closed-form expression and does not require simulations. However, our approach does not rely on normal errors or a specific covariance structure for the error terms within a cluster, unlike those papers.

We perform the analysis within the standard framework where asymptotic normality holds \citep[e.g.,][]{carter2017asymptotic,djogbenou2019asymptotic,hansen2019asymptotic}.
Our Theorem 1 shows that the proposed critical value is asymptotically valid as long as the $t$-statistic is asymptotically normal.
In this sense, our paper is complementary to the recent literature that investigates or relaxes the assumption of asymptotic normality.
For example, \cite{hansen2019asymptotic}  provide a necessary condition for the weak law of large numbers to hold with clustered data while allowing for arbitrary within dependent structures.\footnote{With the notations in Section \ref{sec:proposed_method}, their necessary condition (Assumption 1) is written as $\max_{g}N_g/(N_1+\ldots+N_G)=o(1)$.} Based on their result, asymptotic normality requires that no cluster can have a disproportionately large number of observations compared to the others.  In addition, \cite{kojevnikov2023some} show the impossibility of estimating the asymptotic variance in the presence of a single large cluster.
Finally, \cite{chiang2023genuinely} provide a necessary and sufficient condition for the asymptotic normality of the $t$-statistic. Intuitively, their condition is that the score for the $t$-statistic has a sufficiently thin tail.
In our Theorem 2, we show that for the scores with bounded higher-order moments, our proposed critical value achieves third-order refinement in addition to basic asymptotic size control.

The paper proceeds as follows. Section \ref{sec:proposed_method} formally introduces the regression model with clustered errors and the new critical value. Section \ref{sec:montecarlo} presents Monte Carlo simulations for the proposed critical value and the existing ones. Section \ref{sec:conclusion} concludes. The appendix collects all the proofs and additional results.
 \section[Inverting a Cram\'er--Edgeworth Expansion with Clustered Errors]{Inverting a Cram\'er--Edgeworth Expansion\\with Clustered Errors}\label{sec:proposed_method}

\subsection{The Setup}

We consider the dataset of $\{(Y_{ig},X_{ig}):i=1,\ldots,N_g,\ g=1,\ldots,G\}$ that follows the regression model $Y_{ig}=X_{ig}'\beta+u_{ig}$ with $E[u_{ig}\mid X_{1g},\ldots,X_{N_g g}]=0$ and $\dim(X_{ig})=k$. We assume the observations $\{Y_{ig}:i=1,\ldots,N_g\}$ are independent across $g=1,\ldots,G$ given $\{X_{ig}:i=1,\ldots,N_g\}$. To emphasize that we use independence across $g$, we use matrix notation with $Y_g=(Y_{1g},\ldots,Y_{N_g g})'$ and $X_g=(X_{1g}',\ldots,X_{N_g g}')'$, and write the regression model succinctly as
$$
Y_{g}=X_{g}\beta+u_{g}\mbox{ with }E[u_{g}\mid X_{g}]=0.
$$
For a fixed vector $\lambda\in \mathbb R^{k}$ and a hypothesized value $c_0$ for $\lambda'\beta$, we consider the hypothesis testing problem of
$$
H_0:\lambda'\beta=c_0\mbox{ vs }H_1: \lambda'\beta\ne c_0
$$
with the significance level $\alpha\in (0,1)$.
The OLS estimator for $\beta$ is
$$\hat{\beta}=\left(\sum_{g=1}^GX_g'X_g\right)^{-1}\left(\sum_{g=1}^GX_g' Y_g\right).$$
Throughout this paper, we consider an asymptotic variance estimator for $\lambda'\hat{\beta}$ defined by
$$\hat\sigma^2=\frac{1}{G}\sum_{g=1}^G(\lambda'\Pi X_g'\hat{u}_g)^2$$ with $\Pi=(\frac{1}{G}\sum_{g=1}^GX_{g}'X_{g})^{-1}$ and $\hat{u}_g=Y_g-X_g\hat\beta$.\footnote{Remark \ref{rem:another_var_estimation} discusses another asymptotic variance estimator.}
Define the $t$-statistic by $$t=\sqrt{G}\frac{\lambda'\hat{\beta}-c_0}{\hat\sigma}.$$
From now on, we estimate the null distribution for the above $t$-statistic and construct a critical value. We treat the covariates $\mathbf{X}=\{X_{g}\}^\infty_{g=1}$ as fixed, so we investigate $Pr(|t|\leq z\mid\mathbf{X}=\mathbf{x})$ for a given sequence of constants $\mathbf{x}=\{x_{g}\}^\infty_{g=1}$.

We consider the numerator and denominator of
$$
t=\sqrt{G}\frac{(\lambda'\hat{\beta}-c_0)/\sigma}{\hat\sigma/\sigma}
$$
under the null hypothesis $H_0$, where $\sigma$ is the population counterpart of $\hat\sigma$ defined by $\sigma^2=\frac{1}{G}\sum_{g=1}^G\sigma_g^2$ with $\sigma_{g}^2=E\left[(\lambda'\Pi X_g'u_g)^2\mid\mathbf{X}=\mathbf{x}\right]$. Under the null hypothesis, the numerator has the linear representation of
$$
\frac{\lambda'\hat{\beta}-c_0}{\sigma}=\frac{1}{G}\sum_{g=1}^G\omega_{1g}\mbox{ with } \omega_{1g}=\sigma^{-1}\lambda'\Pi X_g'u_g.
$$

The lemma below shows  the denominator $\hat\sigma/\sigma$ can be expressed as sample means of independent random variables.

\begin{lem}\label{lemma:variace_quadratic}
$$
\frac{\hat\sigma}{\sigma}=\sqrt{1-\left(\frac{1}{G}\sum_{g=1}^G\omega_{2g}\right)'\Gamma\left(\frac{1}{G}\sum_{g=1}^G\omega_{2g}\right)+\frac{1}{G}\sum_{g=1}^G\omega_{3g}}
$$
where
\begin{align*}
\Gamma
&=\left(
\begin{array}{cc}
-\frac{1}{G}\sum_{g=1}^GX_g'X_g\Pi'\lambda\lambda'\Pi X_g'X_g&I_k\\
I_k&O
\end{array}
\right),
\\
\omega_{2g}&=\sigma^{-1}\left(\begin{array}{c}I_k\\ X_g'X_g\Pi'\lambda\lambda'\end{array}\right) \Pi X_g'u_g\mbox{, and }
\\
\omega_{3g}&=\sigma^{-2}((\lambda'\Pi X_g'u_g)^2-\sigma_{g}^2).
\end{align*}
\end{lem}
\begin{proof}
See Appendix \ref{sec:proof_lemma1}.
\end{proof}

\subsection{Main Results}
We can use the (second-order) Cram\'er-Edgeworth expansion for the $t$-statistic's null distribution, for which we provide a sufficient condition in Section \ref{sub:sufficient}.
To introduce the Cram\'er-Edgeworth expansion, we define the following moments of $\omega_{1g}$ and $\omega_{2g}$:
$$
(\mu_{1,2}',\mu_{2,2},\mu_{1,1,1},\mu_{1,1,1,1})
=
\frac{1}{G}\sum_{g=1}^GE\left[(
\omega_{1g}\omega_{2g}',\omega_{2g}'\Gamma\omega_{2g},\omega_{1g}^3,\omega_{1g}^4)\mid\mathbf{X}=\mathbf{x}
\right].
$$
Using these moments, we introduce the second Edgeworth polynomial
$$
q_2(z)=-\left(\frac{1}{2}(k_2+k_1^2)He_1(z)+\frac{1}{24}(k_4+4k_1k_3)He_3(z)+\frac{1}{72}k_3^2He_5(z)\right),
$$
where $He_r$ are the $r$-th order Hermite polynomials, and  $(k_1,k_2,k_3,k_4)=(\nu_1,\nu_2-\nu_1^2,\nu_3-3\nu_1,\nu_4-4\nu_1\nu_3-6\nu_2+12\nu_1^2)$
with
\[
\begin{pmatrix}
\nu_1\\[2pt]
\nu_2\\[2pt]
\nu_3\\[2pt]
\nu_4
\end{pmatrix}
=
\begin{pmatrix}
-\dfrac{\mu_{1,1,1}}{2}\\
2\mu_{1,1,1}^{\,2}
+\mu_{2,2}
+2\mu_{1,2}^{\prime}\Gamma\mu_{1,2}\\
-\dfrac{7}{2}\mu_{1,1,1}\\
-2\mu_{1,1,1,1}
+28\mu_{1,1,1}^{\,2}
+6\mu_{2,2}
+24\mu_{1,2}^{\prime}\Gamma\mu_{1,2}
\end{pmatrix}.
\]
We can then introduce the (second-order) Cram\'er-Edgeworth expansion for the absolute value of the $t$ statistic,
\begin{equation}\label{eq:edgeworth_expansion}
\sup_{z\in\mathcal{N}}\Big|Pr(|t|\leq z\mid\mathbf{X}=\mathbf{x})-(2 \Phi(z) - 1 +2G^{-1}q_2(z)\phi(z))\Big| = o(G^{-1})\mbox{ under }H_0,
\end{equation}
where $\mathcal{N}$ is some neighborhood of $\Phi^{-1}(1-\alpha/2)$.

The function $q_2(z)$ in the (second-order) Cram\'er-Edgeworth expansion is only known up to parameters $\mu_{1,2},\mu_{2,2},\mu_{1,1,1},\mu_{1,1,1,1}$. We estimate them using their sample analogs:
$$
(\hat\mu_{1,2}',\hat\mu_{2,2},\hat\mu_{1,1,1},\hat\mu_{1,1,1,1})
=
\frac{1}{G}\sum_{g=1}^G\left(
\hat\omega_{1g}\hat\omega_{2g}',\hat\omega_{2g}'\Gamma\hat\omega_{2g},\hat\omega_{1g}^3,\hat\omega_{1g}^4\right),
$$
where $(\hat\omega_{1g},\hat\omega_{2g})$ is defined by replacing $(\sigma,u_g)$ with $(\hat\sigma,\hat{u}_g)$ in the expression for $(\omega_{1g},\omega_{2g})$. We  construct the estimator $\hat{q}_2(z)$ for $q_2(z)$ using these sample analogs. We use the Cornish-Fisher expansion (i.e., the inversion of the Cram\'er-Edgeworth expansion) with the estimated value of $\hat{q}_2(z)$. Namely, our \emph{analytical} critical value consists of the usual two-sided normal critical value and $O(G^{-1})$ correction term,
$$
\hat{cv}=\Phi^{-1}(1-\alpha/2)-  G^{-1}\hat q_2(\Phi^{-1}(1-\alpha/2)).
$$
The resulting confidence interval for $\lambda'\beta$ is $\lambda'\hat\beta\pm \hat{cv}\hat\sigma/\sqrt{G}$.

We impose the following conditions on the moments of the data.
\begin{assumption}\label{assumption:moments_and_invertibility}
(i) $\sigma^2$ and the minimum eigenvalue of $\frac{1}{G}\sum_{g=1}^Gx_{g}'x_{g}$ are bounded away from zero uniformly in $G$.(ii) $\left\| x_g'x_g\right\|^4$ is  bounded uniformly in $g$. (iii) $E[\left\| X_g'u_g\right\|^4\mid\mathbf{X}=\mathbf{x}]$ is bounded uniformly in $g$.
\end{assumption}

Even if the Cram\'er-Edgeworth expansion fails, we can still    control size with an asymptotic approximation error of $o(1)$, as long as the $t$ statistic is asymptotically  normal under the null hypothesis.

\begin{thm}\label{thm:theorem_robust}
(i) If the $t$-statistic $t$ converges in distribution to the standard normal distribution and $\hat q_2(\Phi^{-1}(1-\alpha/2))=o_p(G)$, then $Pr(|t|\leq \hat{cv}\mid\mathbf{X}=\mathbf{x})=1-\alpha+o(1)$ under the null hypothesis $H_0$.\footnote{The condition $\hat q_2(\Phi^{-1}(1-\alpha/2))=o_p(G)$ holds, e.g., if one estimates $\mu_{1,2},\mu_{2,2},\mu_{1,1,1},\mu_{1,1,1,1}$ with truncated means using a threshold value of $o(G)$.}
(ii) If Assumption \ref{assumption:moments_and_invertibility} holds, then $t$ converges in distribution to the standard normal distribution under the null hypothesis $H_0$.
\end{thm}
\begin{proof}
See Appendix \ref{sec:proof_theorem_robust}.
\end{proof}

In Theorem \ref{thm:theorem1} below, we show the asymptotic refinement for the proposed critical value.

\begin{thm}\label{thm:theorem1}
Suppose Assumption \ref{assumption:moments_and_invertibility} holds. If $E[\left\| X_g'u_g\right\|^{16}\mid\mathbf{X}=\mathbf{x}]$ is bounded uniformly in $g$ and the Cram\'er-Edgeworth expansion in \eqref{eq:edgeworth_expansion} holds, then $Pr(|t|\leq \hat{cv}\mid\mathbf{X}=\mathbf{x})=1-\alpha+o(G^{-1})$ under the null hypothesis $H_0$.
\end{thm}

\begin{proof}
See Appendix \ref{sec:proof_theorem1}.
\end{proof}

The bounded 16th moment assumption is crucial for our inference because we need to estimate the population objects of $\mu_{1,2}$, $\mu_{2,2}$, $\mu_{1,1,1}$, and $\mu_{1,1,1,1}$ in the Cram\'er-Edgeworth expansion with negligible error (formalized in Lemma \ref{assn:assumption_estimated_q2_rate}).
Note that $\mu_{1,1,1,1}$ is the average fourth moment of $\omega_{1g}$'s and we estimate it by $\hat\mu_{1,1,1,1}$. In the proof, we bound the fourth moment of the estimator and thus require bounded 16th moments of $\omega_{1g}$'s.

Combining the above two theorems, our critical value dominates the normal critical value in the sense that (i) our critical value is asymptotically valid whenever the asymptotic normality of $t$ holds and (ii) it offers an asymptotic refinement in some cases (e.g., when the assumptions in Theorem \ref{thm:theorem1} hold).

\begin{rem}\label{rem:another_var_estimation}
If we use another asymptotic variance estimator for $\lambda'\hat{\beta}$ denoted by $\tilde\sigma^2$, we can apply the same strategy of inverting the Cram\'er-Edgeworth expansion by deriving the expansion of $\sqrt{G}({\lambda'\hat{\beta}-c_0})/{\tilde\sigma}$. Alternatively, we can still use the same confidence interval $\lambda'\hat\beta\pm \hat{cv}\hat\sigma/\sqrt{G}$ even if we use $\tilde\sigma^2$.\footnote{With the asymptotic variance estimator $\tilde\sigma^2$, the $t$-statistic is $\sqrt{G}({\lambda'\hat{\beta}-c_0})/{\tilde\sigma}=({\hat\sigma}/{\tilde\sigma})t$. We can accordingly change the critical value as $(\hat\sigma/\tilde\sigma)\hat{cv}$, but the resulting confidence interval is still the same.}
\end{rem}

\begin{rem}\label{rem:different_edgeworth}
In this paper, we treated $\mathbf{X}$ as fixed, so we do not apply the Cram\'er-Edgeworth expansion to $\frac{1}{G}\sum_{g=1}^GX_{g}'X_{g}$. Our approach is different from \cite{djogbenou2019asymptotic}, who consider the regressor $\mathbf{X}$ as random when they derive the Cram\'er-Edgeworth expansion for the $t$-statistic $t$.\footnote{They do not use the Cram\'er-Edgeworth expansion to construct an asymptotically refined inference, but instead use it to show that the wild cluster bootstrap cannot achieve the asymptotic refinement of $o(G^{-1})$ unless the third cumulant of the score is zero. As a result, this remark is not relevant for their analysis of the wild cluster bootstrap because their conclusion relies on the fact that the wild cluster bootstrap cannot replicate the third cumulant in the Cram\'er-Edgeworth expansion.} First, our expansion does not require the independence of $X_{g}$ across $g$. As a consequence, we allow $X_1,\ldots,X_G$ to be generated from adaptive randomization. Second, the assumption in \cite{djogbenou2019asymptotic} for the Cram\'er-Edgeworth expansion excludes empirically relevant cases, such as binary regressors. With the random regressor, they apply the Cram\'er-Edgeworth expansion to the term $\frac{1}{G}\sum_{g=1}^GX_{g}'X_{g}$ and assume the Cram\'er condition on $X_{g}'X_{g}$. This condition fails if $X_g$ includes discrete variables \citep[p.207]{bhattacharya2010normal}.\footnote{Consider the case where $X_{ig}$ includes a binary variable $D_g$. The matrix $X_g'X_g$ includes $N_gD_g^2$, which takes the two values of $N_g$ and $0$. By choosing the vector $a$ such that $a'\mathrm{vech}(X_{g}'X_{g})=\|a\|N_gD_g^2$, we have $E[\exp(ia'\mathrm{vech}(X_{g}'X_{g}))]=Pr(D_g=1)\exp(i\|a\|N_g)+(1-Pr(D_g=1))=1$ if $\|a\|N_g$ is a multiple of $2\pi$. This implies that the Cram\'er condition fails.} In this paper, we do not require the Cram\'er condition on $X_{g}'X_{g}$.  Third, the resulting Cram\'er-Edgeworth expansion has fewer unknown parameters to estimate. Namely, we do not need to estimate the moments of  $X_{g}'X_{g}$.
\end{rem}

\subsection{Sufficient Conditions for Expansion in Theorem \ref{thm:theorem1}}\label{sub:sufficient}

In this subsection, we provide a sufficient condition for the Cram\'er-Edgeworth expansion in Theorem~\ref{thm:theorem1}. Define the augmented score vector
$$
\tilde\eta_g=
\left(\begin{array}{c}
(1,\mathrm{vec}(X_g'X_g)')'\otimes(X_g'u_g)\\
(X_g'u_g)\otimes(X_g'u_g)
\end{array}\right).
$$
Since some components in $\tilde\eta_g$ could be linearly dependent, we can construct $\eta_g$ by removing the linearly dependent components (cf. Remark \ref{remark:removal}).
This removal is necessary to normalize the random variable $\frac{1}{\sqrt{G}}\sum_{g=1}^G\eta_g$ by using the matrix square root of its variance matrix.

\begin{thm}\label{thm:theorem3}
The Cram\'er-Edgeworth expansion in \eqref{eq:edgeworth_expansion} holds if (i) the eigenvalues of $Var(\eta_g\mid\mathbf{X}=\mathbf{x})$ are bounded away from zero  uniformly in $g$, (ii) $a'\eta_g$ has a pdf that is bounded uniformly in $(g,a)$ with $a\in \mathbb R^{\operatorname{dim}\eta}$, $\|a\|=1$, and (iii) for every positive integer $j$, the $j$th moment of $\eta_g$ is bounded uniformly in $g$.
\end{thm}
\begin{proof}
See Appendix \ref{sec:proof_theorem3}.
\end{proof}

The first assumption imposes a nonzero lower bound on the variance of $\eta_g$ across $g$. The second assumption requires $\eta_g$ to be continuously  distributed. The continuous distribution of $\eta_g$ holds when the distribution of the error term $u_g$ is absolutely continuous (with respect to the $N_g$-dimensional Lebesgue measure). As in the standard arguments for the Cram\'er condition \citep[p.207]{bhattacharya2010normal}, we can impose this condition on a component of the distribution of $\eta_g$, instead of the entire  distribution. The third condition of bounded moments for all orders simplifies the proofs. It is possible to weaken it to the bounded moment condition up to a certain order by truncating $\eta_g$ (cf., \citealp[p.256]{hall1992bootstrap}; \citealp[p.446]{bhattacharya1978validity}).

\begin{rem}\label{remark:removal}
We follow \cite{djogbenou2019asymptotic} when defining $\eta_g$ by removing linearly dependent components from $\tilde\eta_{g}$. Formally, we define $\eta_g$ by using two matrices $\mathbf{M}_1$ and $\mathbf{M}_2$ such that $\tilde\eta_{g}=\mathbf{M}_1\eta_g$ and $\eta_g=\mathbf{M}_2\tilde\eta_{g}$ for every $g$.\footnote{We assume without loss of generality that the first component $\eta_{1g}$ of $\eta_g$ is $\omega_{1g}$.} This definition implicitly assumes the matrix  $Var(\tilde\eta_g\mid\mathbf{X}=\mathbf{x})$ has the same null space across $g$. Since the framework of this paper treats $\mathbf{X}$ as fixed to $\mathbf{x}$, it is possible for the matrix $Var(\tilde\eta_g\mid\mathbf{X}=\mathbf{x})$ to have different null spaces (or different ranks) for different values of $g$. In that case, using an increasing sequence $\{S_\ell\}$ of integers with $S_0=0$, we can consider
$\sum_{g=S_{\ell-1}+1}^{S_{\ell}}\tilde\eta_g$ instead of $\eta_g$ itself,
so that the matrix  $Var(\sum_{g=S_{\ell-1}+1}^{S_{\ell}}\tilde\eta_g\mid\mathbf{X}=\mathbf{x})$ has the same null space across $\ell$.
\end{rem}
 \section{Monte Carlo Simulations}\label{sec:montecarlo}

In this section, we investigate the finite-sample performance of the critical value proposed in Section \ref{sec:proposed_method} using simulated data. We compare our method (denoted by ``Analytical'' in the figures) with a few existing methods, such as (i) the $t_{G-1}$ critical value (``Student''), (ii) the restricted wild cluster bootstrap with Rademacher weights by \cite{cameron2008bootstrap} (``WCB''), and (iii) the pairs percentile-$t$ cluster bootstrap  (``Pairs'') as well as the standard normal critical value (``Normal'').\footnote{When we use the $t_{G-1}$ critical value, we multiply $\hat\sigma^2$ by a finite sample adjustment term: $d_1=\frac{(N-1)G}{(N-k)(G-1)}$, $d_2=\frac{G}{(G-1)}$, and $d_3=\frac{(N-1)G}{(N-k-G)(G-1)}$ with $N=\sum_{g=1}^GN_g$. Option $d_3$ treats each fixed effect as an additional regressor.  We use $d_1$ as the benchmark in the main text. Results for $d_2$ and $d_3$ are provided in Appendix Section \ref{sec:addendix_sim_tables}.

For the pairs cluster bootstrap, we use the Moore-Penrose inverse when we cannot invert a resampled matrix of $\frac{1}{G}\sum_{g=1}^GX_{g}'X_{g}$.}   We also compute the infeasible true Cornish-Fisher expansion as a comparison for our proposed critical value. For all four designs, we use 10,000 simulations; all bootstrap procedures use 1000 draws, and the numbers of clusters are 10, 25, 50, 75, 100, and 200. We consider four designs that are challenging for existing methods but can be accommodated well using our analytical corrections. The first design features binary regressors, which are challenging for the pairs cluster bootstrap, while the second design features skewed errors, which are challenging for the wild cluster bootstrap. The third design combines these two features. The fourth design is a generic fixed-effects panel design with unequal cluster sizes, a binary within-cluster regressor, and skewed errors.

In all the designs, the normal critical value $\Phi^{-1}(1-\alpha/2)$ over-rejects the null hypothesis even when it is true. The additional term of $-G^{-1}\hat q_2(\Phi^{-1}(1-\alpha/2))$ in our proposed critical value reduces the magnitude of over-rejections. The comparison with other methods depends on the design. Especially, the skewness $\mu_{1,1,1}$ affects the comparison between our proposal and the wild cluster bootstrap.

\begin{figure}[t]
\centering
\begin{subfigure}[t]{0.48\linewidth}
\centering
\includegraphics[width=\linewidth,keepaspectratio]{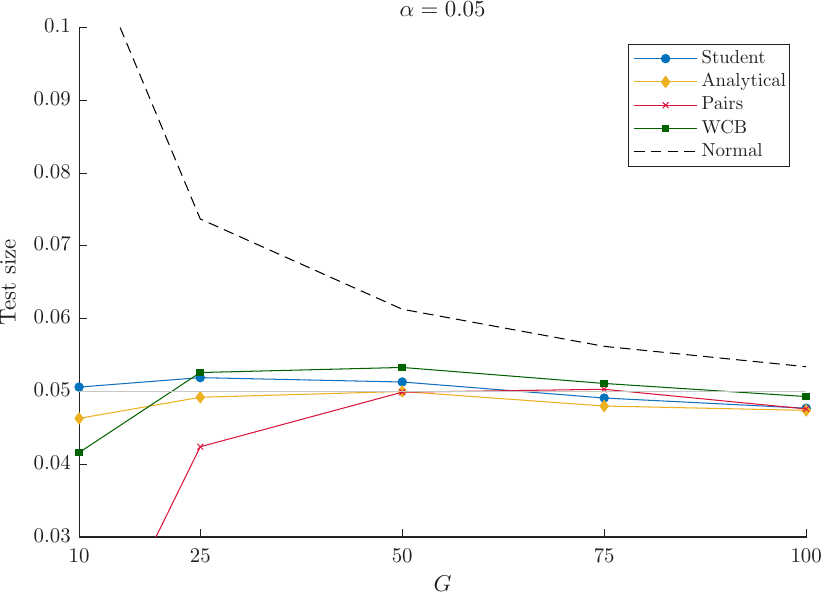}
\caption{Test size}
\end{subfigure}\hfill
\begin{subfigure}[t]{0.48\linewidth}
\centering
\includegraphics[width=\linewidth,keepaspectratio]{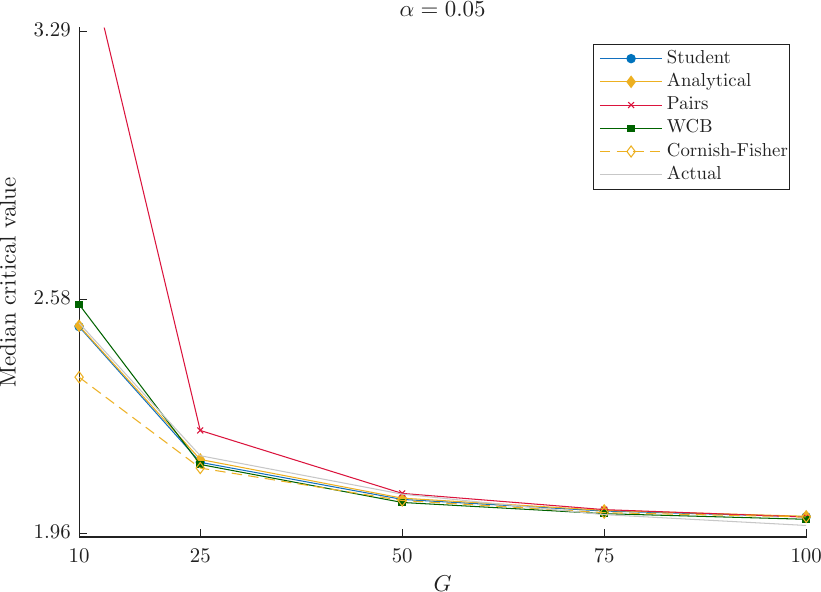}
\caption{Median critical value}
\end{subfigure}
\caption{Test size and median critical values for two-sided tests in the \cite{bertrand2004much} design. The left panel reports test size, and the right panel reports median critical values. $\alpha$ is the nominal test size. Panel (b) vertical grid values correspond to two-sided Normal critical values for tests with nominal sizes 5\%, 1\%, and 0.1\%.}
\label{fig:figure_BDM}
\end{figure}
\begin{figure}[t]
\centering
\begin{subfigure}[t]{0.48\linewidth}
\centering
\includegraphics[width=\linewidth,keepaspectratio]{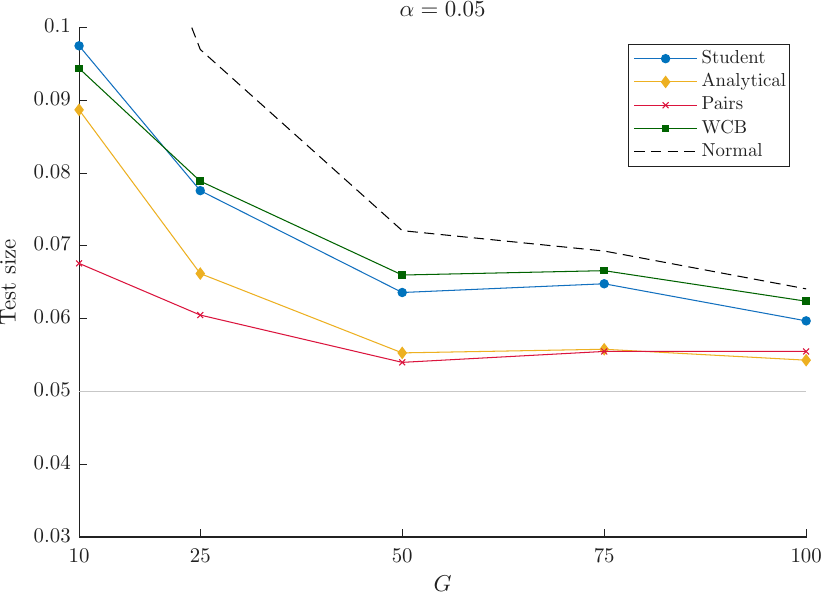}
\caption{Test size}
\end{subfigure}\hfill
\begin{subfigure}[t]{0.48\linewidth}
\centering
\includegraphics[width=\linewidth,keepaspectratio]{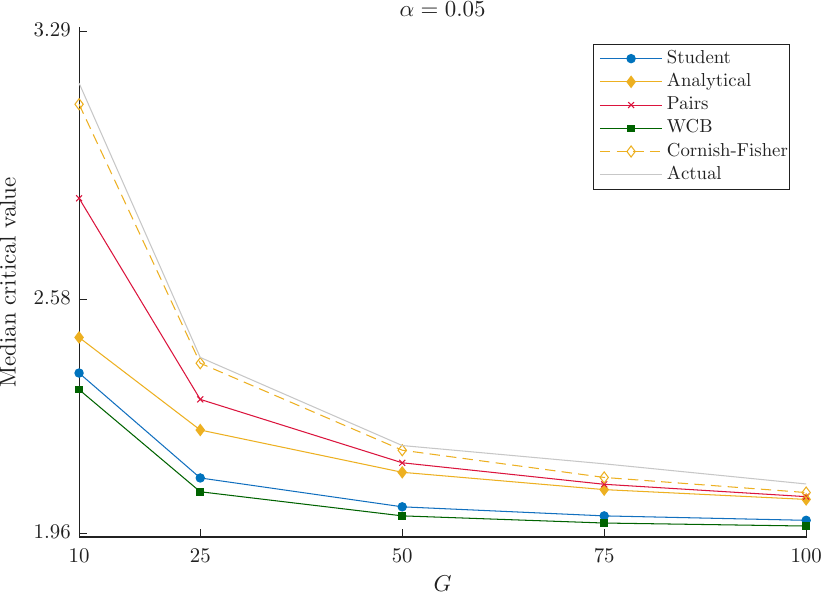}
\caption{Median critical value}
\end{subfigure}
\caption{Test size and median critical values for two-sided tests with a skewed error distribution. The left panel reports test size, and the right panel reports median critical values.  $\alpha$ is the nominal test size. Panel (b) vertical grid values correspond to two-sided Normal critical values for tests with nominal sizes  5\%, 1\%, and 0.1\%.}
\label{fig:figure_SkewedScore}
\end{figure}
\begin{figure}[t]
\centering
\begin{subfigure}[t]{0.48\linewidth}
\centering
\includegraphics[width=\linewidth,keepaspectratio]{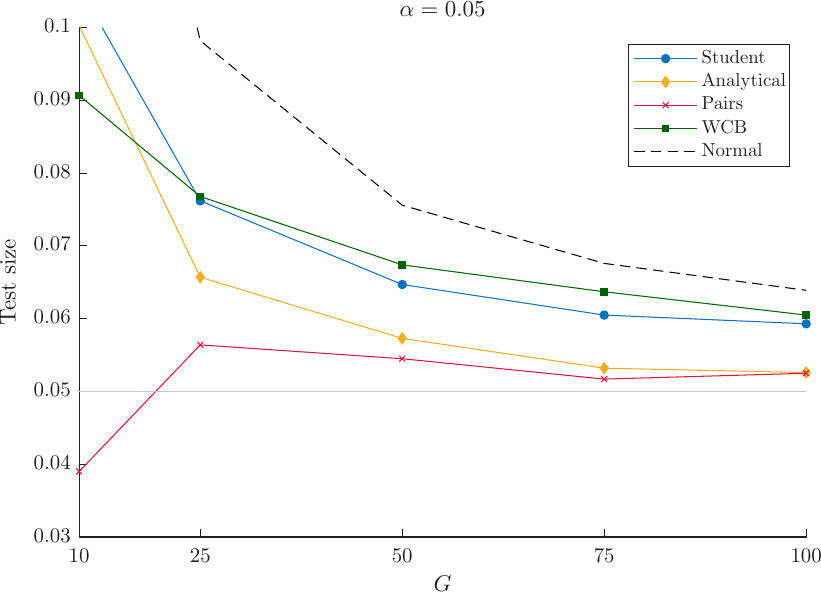}
\caption{Test size}
\end{subfigure}\hfill
\begin{subfigure}[t]{0.48\linewidth}
\centering
\includegraphics[width=\linewidth,keepaspectratio]{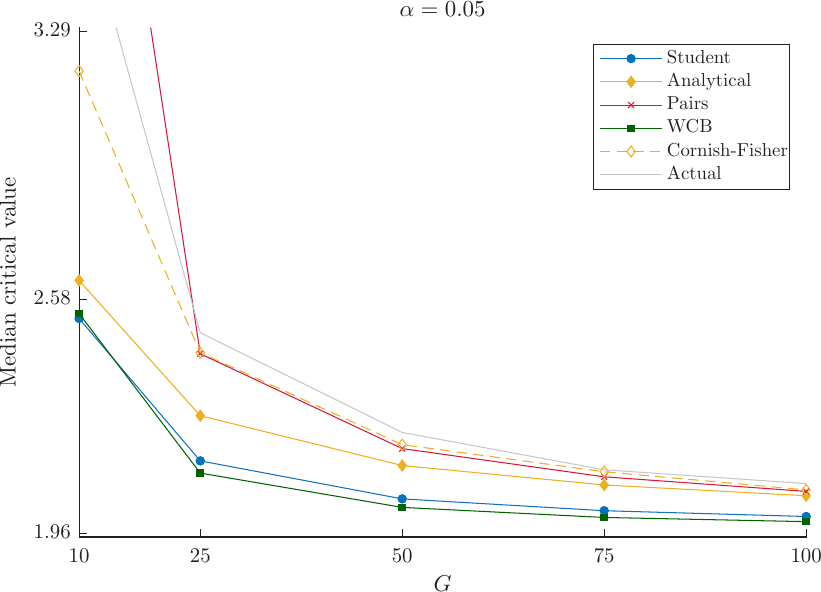}
\caption{Median critical value}
\end{subfigure}
\caption{Test size and median critical values for two-sided tests with a skewed error and a binary regressor. The left panel reports test size, and the right panel reports median critical values.   $\alpha$ is the nominal test size. Panel (b) vertical grid values correspond to two-sided Normal critical values for tests with nominal sizes  5\%, 1\%, and 0.1\%.}
\label{fig:figure_combined}
\end{figure}
\begin{figure}[t]
\centering
\begin{subfigure}[t]{0.48\linewidth}
\centering
\includegraphics[width=\linewidth,keepaspectratio]{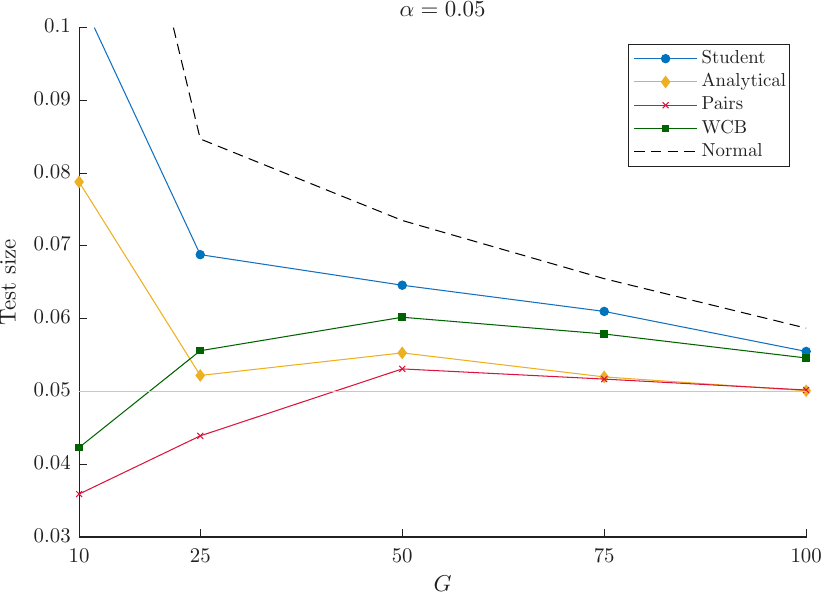}
\caption{Test size}
\end{subfigure}\hfill
\begin{subfigure}[t]{0.48\linewidth}
\centering
\includegraphics[width=\linewidth,keepaspectratio]{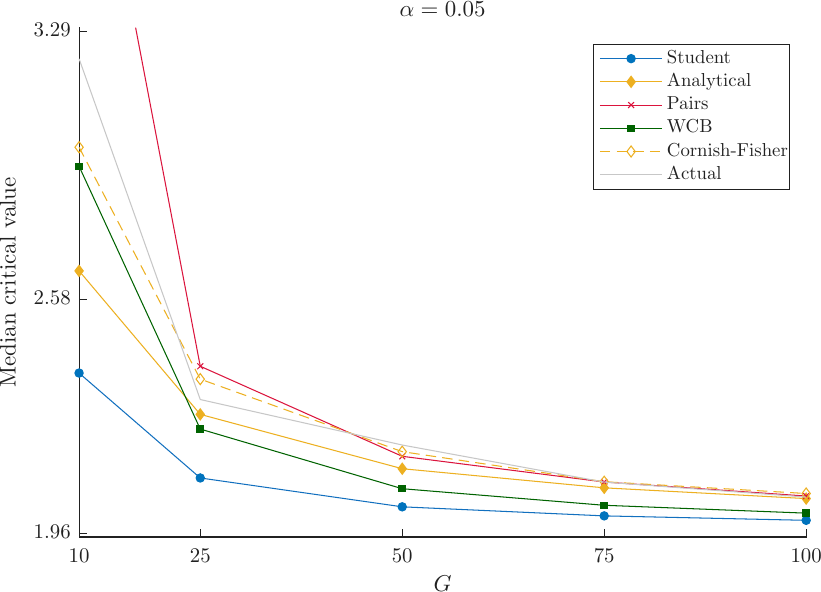}
\caption{Median critical value}
\end{subfigure}
\caption{Test size and median critical values for two-sided tests in the generic fixed-effects design. The left panel reports test size, and the right panel reports median critical values. $\alpha$ is the nominal test size. Panel (b) vertical grid values correspond to two-sided Normal critical values for tests with nominal sizes  5\%, 1\%, and 0.1\%.}
\label{fig:figure_GenericFE}
\end{figure}

\subsection{Design with Binary Regressors}\label{sec:BDM}
In this section, we follow the simulation design (BDM) from \cite{bertrand2004much}  and \citet[Section 5.A]{cameron2008bootstrap}. It uses a state-year panel of excess earnings from 1979 to 1999 based on the Current Population Survey.\footnote{We use the data from the replication package of \cite{cameron2015practitioner}: \url{https://cameron.econ.ucdavis.edu/research/papers.html} In particular, this simulation exercise uses the variable \texttt{lnwage} from \texttt{CPS\_panel.dta} from 1979 to 1999.} For each simulation draw, we randomly select $G$ out of 50 states with replacement. We randomly select the policy change time uniformly from $\{1984,\ldots,1993\}$ and assume that half of the $G$ states experience the policy change after the selected time period. We construct the policy dummy variable accordingly. By definition, this policy dummy variable has a zero coefficient in the population. We regress the excess earnings on the policy dummy,  the year dummies, and the state dummies. Then we conduct the significance test for the coefficient of the policy dummy variable.

In the data generating process of this section, the skewness of the score is close to zero (with $\hat\mu_{1,1,1}=0.02$ for $G=10^4$). Although the Cram\'er condition does not hold for this design because all the variables are discrete \citep[cf.,][Ch.5]{bhattacharya2010normal}, we can consider whether a critical value accounts for the skewness and kurtosis of the $t$-statistic, which are the key components of the second-order Cram\'er-Edgeworth expansion. Our proposed method matches these moments by estimating them explicitly. At the same time, the wild cluster bootstrap approximates the skewness and kurtosis well because it uses zero skewness and consistently estimates the kurtosis \citep[Section 5.2]{djogbenou2019asymptotic}.

Figure \ref{fig:figure_BDM} shows the rejection probabilities for different methods.\footnote{In Appendix \ref{sec:addendix_sim_tables}, we provide the tables that report the rejection probabilities.} As documented in \cite{cameron2008bootstrap}, the pairs cluster bootstrap under-rejects for small values of $G$ (e.g., $G=10$), while the inference based on the $t_{G-1}$ critical value over-rejects. All methods control the size approximately when $G$ is sufficiently large (e.g., $G=50$). Our proposed method exhibits comparable performance to the wild cluster bootstrap, even for a small value of $G=10$.

\subsection{Design with a Skewed Error Distribution}\label{sec:skewed_MONTECARLO_DGP}
To compare the methods when the error has large skewness, we consider the case where $N_g=1$, $X_{ig}=1$, and $Y_{ig}$ follow the demeaned exponential distribution with unit variance. This design is used in Section 3 of \cite{hall1983inverting} for one-sided tests. Since the error has a skewness of $2$, Theorem 5.2 of \cite{djogbenou2019asymptotic} implies that the wild cluster bootstrap does not have asymptotic refinement in this case. Then we conduct the test for the mean being equal to zero.

Figure \ref{fig:figure_SkewedScore} shows the rejection probabilities for different methods with the skewed error distribution. Again, the rejection probabilities of all the methods approach the prespecified significance level ($5\%$) as $G$ increases, which confirms their asymptotic validity. However, the finite-sample performance differs. The rejection probabilities of our proposed method approach $5\%$ faster than those of the wild cluster bootstrap. This is consistent with the theoretical fact that the wild cluster bootstrap does not have asymptotic refinement for this data distribution. In contrast, our analytical approach and the pairs cluster bootstrap both achieve asymptotic refinement. This explains why these methods are similar to each other and much closer to the nominal size than the other two methods in Figure \ref{fig:figure_SkewedScore} as $G$ increases.

\subsection{Design with a Skewed Error Distribution and a Binary Regressor}\label{sec:GenericMC}
Now we consider the case with $N_g=1$, $X_g=(1,1)'$ for the first 50\% of $g$ and $(1,0)'$ for the others, and $Y_g = X_g'\beta+(2X_{g2}-1)u_g$ where $\beta=(0,0)'$ and $u_g$ are the demeaned exponential random variables with unit variance. We conduct the significance test for the second component of $\beta$. In this case, we have $\mu_{1,1,1}=2.00$.This design has two features. First, the regressor is binary. As a result, when we use the pairs cluster bootstrap, the resampled value of $\frac{1}{G}\sum_{g=1}^GX_g'X_g$ can be (nearly) singular. As in \cite{cameron2008bootstrap}, this feature may cause the poor size control of the pairs cluster bootstrap. Second, the error distribution is exponential and therefore the skewness of the score $\mu_{1,1,1}$ is non-zero. In such a case, the asymptotic refinement happens for our proposed method but not for the wild cluster bootstrap.

Figure \ref{fig:figure_combined} shows the rejection probabilities using this data generating process. The results confirm the consequences of the above two features. The pairs cluster bootstrap over-rejects when the number of clusters is small. For $G=25$ or larger, we can see our method is closer to $\alpha=5\%$ than the $t_{G-1}$ critical value or the wild cluster bootstrap.

\subsection{Design with Fixed Effects, Uneven Clusters, and Skewed Errors}\label{sec:generic_FE_MONTECARLO_DGP}
Finally, we consider a design with a scalar regressor and a cluster fixed effect. We set $\beta=0$ and conduct a two-sided test of $H_0:\beta=0$. For a given number of clusters $G$, the cluster sizes are unequal and generated as $$N_g=2+\left[2G\exp(g/G)/(\exp(1/G)+\cdots+\exp(G/G))\right]$$ following \citet[Section 4]{djogbenou2019asymptotic}.\footnote{We add 2 in $N_g$, so the binary regressor has a nonzero variation within each cluster even after the within transformation.} We have $N=\sum_{g=1}^GN_g$ observations in total and index the stacked observations by $j$ with $j=N_1+\cdots+N_{g-1}+i$ for observation $ig$. The regressor is binary and is generated as
\begin{align*}
X_{ig}
=
1\left\{j<N/2\mbox{ and }j\text{ is odd}\right\}.
\end{align*}
For each cluster, we draw a fixed effect $\epsilon_g\sim \mathrm{Uniform}(0.5,1)$.\footnote{The simulation results in this section do not change even if we change the distribution of $\epsilon_g$, since the within transformation removes $\epsilon_g$.} Let $\xi_{ig}$ follow the demeaned exponential distribution with unit variance. The outcome is generated by
\begin{align*}
Y_{ig}
=
\epsilon_g+X_{ig}\beta+(2X_{ig}-1)\xi_{ig}.
\end{align*}
The sign of the skewed error, therefore, depends on the value of the binary regressor.

Figure \ref{fig:figure_GenericFE} shows the rejection probabilities using this data generating process. The rejection probabilities of all methods approach the nominal size as $G$ increases. For small values of $G$, the normal and Student critical values over-reject, while the pairs cluster bootstrap can be conservative. The analytical correction is close to the nominal size once $G$ is moderately large and remains comparable to the pairs cluster bootstrap procedure in this fixed-effects design.
 \section{Conclusion}\label{sec:conclusion}

In this paper, we propose an inference method for linear regression with clustered errors that  achieves third-order asymptotic refinements. Unlike the pairs cluster bootstrap, it does not resample the Gram matrix of $\frac{1}{G}\sum_{g=1}^GX_g'X_g$, thus avoiding a small-sample issue of the pairs cluster bootstrap \citep[cf.,][]{cameron2008bootstrap}. Our simulation results show favorable finite-sample performance of the proposed method, especially when compared with  the conventional normal critical values.
Notably, it works comparably to the wild cluster bootstrap in the simulation design based on \cite{bertrand2004much}, and for some designs with skewed distributions, it has better size control than the wild cluster bootstrap. It supports the theoretical comparison that the wild cluster bootstrap does not achieve asymptotic refinement with skewed errors, while the proposed method does.
 
\appendix
\section{Proofs}

\subsection{Proof of Lemma \ref{lemma:variace_quadratic}}\label{sec:proof_lemma1}

\begin{proof}[Proof of Lemma \ref{lemma:variace_quadratic}]
Note that $\lambda'\Pi X_g'\hat{u}_g=\lambda'\Pi X_g'u_g-\lambda'\Pi X_g'X_g(\hat\beta-\beta)$ and that
\begin{align*}
\frac{1}{G}\sum_{g=1}^G(\lambda'\Pi X_g'\hat{u}_g)^2
&=
\frac{1}{G}\sum_{g=1}^G(\lambda'\Pi X_g'u_g)^2
\\&-2\frac{1}{G}\sum_{g=1}^G(\lambda'\Pi X_g'u_g)'\lambda'\Pi X_g'X_g(\hat\beta-\beta)
\\&+\frac{1}{G}\sum_{g=1}^G(\hat\beta-\beta)'X_g'X_g\Pi\lambda\lambda'\Pi X_g'X_g(\hat\beta-\beta)
\\
&=
\frac{1}{G}\sum_{g=1}^G(\lambda'\Pi X_g'u_g)^2
\\&-2\left(\frac{1}{G}\sum_{g=1}^G(X_g'X_g\Pi\lambda\lambda'\Pi X_g'u_g)\right)'(\hat\beta-\beta)
\\&+\frac{1}{G}(\hat\beta-\beta)'\left(\sum_{g=1}^GX_g'X_g\Pi\lambda\lambda'\Pi X_g'X_g\right)(\hat\beta-\beta).
\end{align*}
Therefore,
\begin{align*}
\hat\sigma^2/\sigma^2
&=
1+\frac{1}{G}\sum_{g=1}^G\omega_{3g}
-2\sigma^{-2}\left(\frac{1}{G}\sum_{g=1}^GX_g'X_g\Pi\lambda\lambda'\Pi X_g'u_g\right)'\left(\frac{1}{G}\sum_{g=1}^G\Pi X_g'u_g\right)
\\&+\sigma^{-2}\left(\frac{1}{G}\sum_{g=1}^G\Pi X_g'u_g\right)'\frac{1}{G}\sum_{g=1}^GX_g'X_g\Pi\lambda\lambda'\Pi X_g'X_g\left(\frac{1}{G}\sum_{g=1}^G\Pi X_g'u_g\right).
\end{align*}
\end{proof}

\subsection{Proof of Theorem \ref{thm:theorem_robust}}\label{sec:proof_theorem_robust}

\begin{proof}[Proof of Theorem \ref{thm:theorem_robust}]
By Slutsky's theorem, it suffices to show $\hat{cv}\rightarrow_p\Phi^{-1}(1-\alpha/2)$ and $t\rightarrow_dN(0,1)$. The first statement follows from the assumption of $\hat q_2(\Phi^{-1}(1-\alpha/2))=o_p(G)$.
The numerator of $t$ is $\frac{1}{\sqrt{G}}\sum_{g=1}^G\omega_{1g}$, which converges in distribution to the standard normal distribution by  Lyapunov central limit theorem with bounded fourth average moment. By Lemma \ref{lemma:variace_quadratic}, the square of the denominator $\hat\sigma^2/\sigma^2$ converges in probability to one, because $\frac{1}{G}\sum_{g=1}^G\omega_{2g}=o_p(1)$ and $\frac{1}{G}\sum_{g=1}^G\omega_{3g}=o_p(1)$ by the weak law of large numbers with bounded second average moment.
\end{proof}

\subsection{Proof of Theorem \ref{thm:theorem1}}\label{sec:proof_theorem1}

\begin{lem}\label{lemma:fourth_moment_markov}
Let $\zeta_1,\ldots,\zeta_G$ be any independent random variables with a bounded average fourth central  moment given $\mathbf{X}=\mathbf{x}$.
For any sequence $\{\varepsilon_G\}$ with $\varepsilon_G>0$ and $G\varepsilon_G^4\rightarrow\infty$, we have
$$
Pr\left(\left|\frac{1}{G}\sum_{g=1}^G(\zeta_g-E[\zeta_g\mid\mathbf{X}=\mathbf{x}])\right|>\varepsilon_G\mid\mathbf{X}=\mathbf{x}\right)
=o(G^{-1}).
$$
\end{lem}
\begin{proof}
By Markov's inequality (for the fourth moment), we have
\begin{align*}
&
Pr\left(\left|\frac{1}{G}\sum_{g=1}^G(\zeta_g-E[\zeta_g\mid\mathbf{X}=\mathbf{x}])\right|>\varepsilon_G\mid\mathbf{X}=\mathbf{x}\right)
\\
&\leq
\frac{E\left[\left(\frac{1}{G}\sum_{g=1}^G(\zeta_g-E[\zeta_g\mid\mathbf{X}=\mathbf{x}])\right)^4\mid\mathbf{X}=\mathbf{x}\right]}{\varepsilon_G^4}
\\
&=
\frac{\frac{1}{G^4}\sum_{g=1}^GE[(\zeta_{g}-E[\zeta_{g}\mid\mathbf{X}=\mathbf{x}])^4\mid\mathbf{X}=\mathbf{x}]}{\varepsilon_G^4}\\&+3\frac{\frac{1}{G^4}\sum_{g_1\ne g_2}E[(\zeta_{g_1}-E[\zeta_{g_1}\mid\mathbf{X}=\mathbf{x}])^2(\zeta_{g_2}-E[\zeta_{g_2}\mid\mathbf{X}=\mathbf{x}])^2\mid\mathbf{X}=\mathbf{x}]}{\varepsilon_G^4},
\end{align*}
where the equality follows from the independence across $g$. Note that
\begin{align*}
&\frac{1}{G^4}\sum_{g_1=1}^G\sum_{g_2=1}^GE[(\zeta_{g_1}-E[\zeta_{g_1}\mid\mathbf{X}=\mathbf{x}])^2(\zeta_{g_2}-E[\zeta_{g_2}\mid\mathbf{X}=\mathbf{x}])^2\mid\mathbf{X}=\mathbf{x}]
\\
&\leq
\frac{1}{G^4}\sum_{g_1=1}^G\sum_{g_2=1}^GE[(\zeta_{g_1}-E[\zeta_{g_1}\mid\mathbf{X}=\mathbf{x}])^4\mid\mathbf{X}=\mathbf{x}]^{1/2}E[(\zeta_{g_2}-E[\zeta_{g_2}\mid\mathbf{X}=\mathbf{x}])^4\mid\mathbf{X}=\mathbf{x}]^{1/2}
\\
&=
\frac{1}{G^2}(\frac{1}{G}\sum_{g=1}^GE[(\zeta_{g}-E[\zeta_{g}\mid\mathbf{X}=\mathbf{x}])^4\mid\mathbf{X}=\mathbf{x}]^{1/2})^2
\\
&\leq
\frac{1}{G^2}\frac{1}{G}\sum_{g=1}^GE[(\zeta_{g}-E[\zeta_{g}\mid\mathbf{X}=\mathbf{x}])^4\mid\mathbf{X}=\mathbf{x}],
\end{align*}
where the second inequality follows from the Cauchy-Schwarz inequality, and the last inequality follows from Jensen's inequality. Since $G\varepsilon_G^4\rightarrow\infty$, we have the statement of this lemma.
\end{proof}

\begin{lem}\label{lemma:polynomials_epsilon_G_convergence}
Let $\hat\vartheta_{1,G},\ldots,\hat\vartheta_{L,G}$ be $L$ sequences of random variables. Suppose $\hat\rho_G$ is the composition of a finite number of matrix additions and/or matrix multiplications of $\hat\vartheta_{1,G},\ldots,\hat\vartheta_{L,G}$. Suppose that the sequence  $\sum_{\ell=1,\ldots,L}\|\vartheta_{\ell,G}\|$ is bounded and $\sum_{\ell=1}^LPr(\|\hat\vartheta_{\ell,G}-\vartheta_{\ell,G}\|>\varepsilon_G\mid\mathbf{X}=\mathbf{x})=o(G^{-1})$ for some diminishing sequence $\{\varepsilon_G\}$. Let $\rho_G$ be the population counterpart of $\hat\rho_G$ with replacing $(\hat\vartheta_{1,G},\ldots,\hat\vartheta_{L,G})$ with $(\vartheta_{1,G},\ldots,\vartheta_{L,G})$. Then there is some constant $C$ such that $$Pr(\|\hat\rho_G-\rho_G\|>C\varepsilon_G\mid\mathbf{X}=\mathbf{x})=o(G^{-1}).$$ If, in addition, $\rho_G$ is an invertible matrix whose singular values are uniformly bounded away from zero, then there is some constant $C$ such that $$Pr(\|\hat\rho_G^{-1}-\rho_G^{-1}\|>C\varepsilon_G\mid\mathbf{X}=\mathbf{x})=o(G^{-1}).$$
\end{lem}
\begin{proof}
For the first result, it suffices to show that $Pr(\|(\hat\vartheta_{1,G}\hat\vartheta_{2,G}+\hat\vartheta_{3,G})-(\vartheta_{1,G}\vartheta_{2,G}+\vartheta_{3,G})\|>\varepsilon_G\mid\mathbf{X}=\mathbf{x})=o(G^{-1})$ for some constant $C$.
Since
\begin{align*}
\|(\hat\vartheta_{1,G}\hat\vartheta_{2,G}+\hat\vartheta_{3,G})-(\vartheta_{1,G}\vartheta_{2,G}+\vartheta_{3,G})\|
&\leq
\|\hat\vartheta_{1,G}-\vartheta_{1,G}\|\|\hat\vartheta_{2,G}-\vartheta_{2,G}\|
+
\|\vartheta_{1,G}\|\|\hat\vartheta_{2,G}-\vartheta_{2,G}\|
\\&\quad +
\|\hat\vartheta_{1,G}-\vartheta_{1,G}\|\|\vartheta_{2,G}\|
+
\|\hat\vartheta_{3,G}-\vartheta_{3,G}\|,
\end{align*}
we have
\begin{align*}
&
Pr(\|(\hat\vartheta_{1,G}\hat\vartheta_{2,G}+\hat\vartheta_{3,G})-(\vartheta_{1,G}\vartheta_{2,G}+\vartheta_{3,G})\|>\varepsilon_G^2+(\|\vartheta_{1,G}\|+\|\vartheta_{2,G}\|)\varepsilon_G+\varepsilon_G\mid\mathbf{X}=\mathbf{x})
\\
&\leq
Pr(\|\hat\vartheta_{1,G}-\vartheta_{1,G}\|\|\hat\vartheta_{2,G}-\vartheta_{2,G}\|>\varepsilon_G^2\mid\mathbf{X}=\mathbf{x})
\\&\quad+
Pr(\|\vartheta_{1,G}\|\|\hat\vartheta_{2,G}-\vartheta_{2,G}\|>\|\vartheta_{1,G}\|\varepsilon_G\mid\mathbf{X}=\mathbf{x})
+
Pr(\|\hat\vartheta_{1,G}-\vartheta_{1,G}\|\|\vartheta_{2,G}\|>\|\vartheta_{2,G}\|\varepsilon_G\mid\mathbf{X}=\mathbf{x})
\\&\quad+
Pr(\|\hat\vartheta_{3,G}-\vartheta_{3,G}\|>\varepsilon_G\mid\mathbf{X}=\mathbf{x})
\\
&\leq
Pr(\|\hat\vartheta_{1,G}-\vartheta_{1,G}\|>\varepsilon_G\mid\mathbf{X}=\mathbf{x})
+
Pr(\|\hat\vartheta_{2,G}-\vartheta_{2,G}\|>\varepsilon_G\mid\mathbf{X}=\mathbf{x})
\\&\quad+
Pr(\|\hat\vartheta_{2,G}-\vartheta_{2,G}\|>\varepsilon_G\mid\mathbf{X}=\mathbf{x})
+
Pr(\|\hat\vartheta_{1,G}-\vartheta_{1,G}\|>\varepsilon_G\mid\mathbf{X}=\mathbf{x})
\\&\quad+
Pr(\|\hat\vartheta_{3,G}-\vartheta_{3,G}\|>\varepsilon_G\mid\mathbf{X}=\mathbf{x})
\\
&=o(G^{-1}).
\end{align*}
By taking $C$ such that $C>\varepsilon_G+(\|\vartheta_{1,G}\|+\|\vartheta_{2,G}\|)+1$, the first result holds. Now we will show the second result about the inverse as follows. By \citet[Eq.(5.8.6)]{horn2012matrix}, we have
$$
\left\|\hat\rho_{G}^{-1}-\rho_{G}^{-1}\right\|
\leq
\|\rho_{G}^{-1}\|
\frac{\|\rho_{G}^{-1}\|\|\hat\rho_{G}-\rho_{G}\|}{1-\|\rho_{G}^{-1}\|\|\hat\rho_{G}-\rho_{G}\|}
$$
as long as $\|\rho_{G}^{-1}\|\|\hat\rho_{G}-\rho_{G}\|<1$. When $\varepsilon_G\leq 0.5\|\rho_{G}^{-1}\|^{-1}$, we have
$$
\|\hat\rho_{G}-\rho_{G}\|<\varepsilon_G
\implies
\left\|\hat\rho_{G}^{-1}-\rho_{G}^{-1}\right\|
\leq
\|\rho_{G}^{-1}\|
\frac{\|\rho_{G}^{-1}\|\varepsilon_G}{1-\|\rho_{G}^{-1}\|\varepsilon_G}\leq
2\|\rho_{G}^{-1}\|^2\varepsilon_G.
$$
Therefore, the second result holds.
\end{proof}

\begin{lem}\label{assn:assumption_estimated_q2_rate}
Suppose the assumptions in Theorem \ref{thm:theorem1} as well as the null hypothesis $H_0$. For every $z>0$, there is some sequence $\{\varepsilon_G\}$ such that $\varepsilon_G=o(1)$ and that $Pr(|\hat q_2(z)-q_2(z)|>\varepsilon_G\mid\mathbf{X}=\mathbf{x})=o(G^{-1})$.
\end{lem}
\begin{proof}
By Lemma \ref{lemma:polynomials_epsilon_G_convergence} and the Cram\'er-Edgeworth expansion in \eqref{eq:edgeworth_expansion}, it suffices to investigate $\hat\mu_{1,2}-\mu_{1,2}$, $\hat\mu_{2,2}-\mu_{2,2}$, $\hat\mu_{1,1,1}-\mu_{1,1,1}$, and  $\hat\mu_{1,1,1,1}-\mu_{1,1,1,1}$. Note that
\begin{align*}
\hat\mu_{1,2}-\mu_{1,2}
&=
\frac{1}{G}\sum_{g}((\hat\omega_{1g}-\omega_{1g})\omega_{2g}+\omega_{1g}(\hat\omega_{2g}-\omega_{2g})+(\hat\omega_{1g}-\omega_{1g})(\hat\omega_{2g}-\omega_{2g}))\\&+\frac{1}{G}\sum_{g}(\omega_{1g}\omega_{2g}-E[\omega_{1g}\omega_{2g}\mid\mathbf{X}=\mathbf{x}])
\\
\hat\mu_{2,2}-\mu_{2,2}
&=
\frac{1}{G}\sum_{g}((\hat\omega_{2g}-\omega_{2g})'\Gamma\omega_{2g}+\omega_{2g}'\Gamma(\hat\omega_{2g}-\omega_{2g})+(\hat\omega_{2g}-\omega_{2g})'\Gamma(\hat\omega_{2g}-\omega_{2g}))\\&+\frac{1}{G}\sum_{g}(\omega_{2g}'\Gamma\omega_{2g}-E[\omega_{2g}'\Gamma\omega_{2g}\mid\mathbf{X}=\mathbf{x}]),
\\
\hat\mu_{1,1,1}-\mu_{1,1,1}
&=
\frac{1}{G}\sum_{g}(\hat\omega_{1g}-\omega_{1g})((\hat\omega_{1g}-\omega_{1g})^2+3(\hat\omega_{1g}-\omega_{1g})\omega_{1g}+3\omega_{1g}^2)\\&+\frac{1}{G}\sum_{g}(\omega_{1g}^3-E[\omega_{1g}^3\mid\mathbf{X}=\mathbf{x}]),
\\
\hat\mu_{1,1,1,1}-\mu_{1,1,1,1}
&=
\frac{1}{G}\sum_{g}(\hat\omega_{1g}-\omega_{1g})(4\omega_{1g}^3+ 6\omega_{1g}^2(\hat\omega_{1g}-\omega_{1g}) + 4\omega_{1g}(\hat\omega_{1g}-\omega_{1g})^2 + (\hat\omega_{1g}-\omega_{1g})^3)\\&+\frac{1}{G}\sum_{g}(\omega_{1g}^4-E[\omega_{1g}^4\mid\mathbf{X}=\mathbf{x}]),
\end{align*}
\begin{align*}
\hat\omega_{1g}-\omega_{1g}&=((\hat\sigma/\sigma)^{-1}-1)\omega_{1g}-(\hat\sigma/\sigma)^{-1}\sigma^{-1}\lambda'\Pi X_g'X_g(\hat\beta-\beta),
\\
\hat\omega_{2g}-\omega_{2g}&=((\hat\sigma/\sigma)^{-1}-1)\omega_{2g}
-
(\hat\sigma/\sigma)^{-1}\sigma^{-1}\left(\begin{array}{c}I_k\\ X_g'X_g\Pi'\lambda\lambda'\end{array}\right)\Pi X_g'X_g(\hat\beta-\beta).
\end{align*}
Each term in the above expressions is bounded as follows:
$$
\|\hat\beta-\beta\|\leq\left\|\Pi\right\|
\left\|\frac{1}{G}\sum_{g=1}^GX_g' u_g\right\|,
$$
\begin{align*}
|\hat\omega_{1g}-\omega_{1g}|&\leq|(\hat\sigma/\sigma)^{-1}-1||\omega_{1g}|\\&+\|\hat\beta-\beta\||\hat\sigma/\sigma|^{-1}\sigma^{-1}\|\lambda\|\|\Pi\|\|X_g'X_g\|,
\\
\|\hat\omega_{2g}-\omega_{2g}\|&\leq|(\hat\sigma/\sigma)^{-1}-1|\|\omega_{2g}\|\\&
+
\|\hat\beta-\beta\||\hat\sigma/\sigma|^{-1}\sigma^{-1}\|\Pi\|\|X_g'X_g\|\\&+\|\hat\beta-\beta\||\hat\sigma/\sigma|^{-1}\sigma^{-1}\|\lambda\|^2\|\Pi\|^2\|X_g'X_g\|^2,
\end{align*}
\begin{align*}
|\omega_{1g}|&\leq\sigma^{-1}\|\lambda\|\|\Pi\|\|X_g'u_g\|,\mbox{ and }\\
\|\omega_{2g}\|&\leq\sigma^{-1}\|\Pi\|\|X_g'u_g\|+\sigma^{-1}\|\lambda\|^2\|\Pi\|^2\|X_g'X_g\|\|X_g'u_g\|.
\end{align*}
By Lemma \ref{lemma:fourth_moment_markov}-\ref{lemma:polynomials_epsilon_G_convergence} and Assumption \ref{assumption:moments_and_invertibility}, we have the statement of this lemma.
\end{proof}

\begin{proof}[Proof of Theorem \ref{thm:theorem1}]
This proof essentially follows the proof of  \citet[Theorem 1]{hall1983inverting}, but it makes sure each step does not use the identical distribution assumption in \cite{hall1983inverting}.
Let $z=\Phi^{-1}(1-\alpha/2)$. Using the sequence $\varepsilon_G$ in Lemma \ref{assn:assumption_estimated_q2_rate}, we have
\begin{align}
Pr(|t|\leq z -  G^{-1}\hat q_2(z) \mid\mathbf{X}=\mathbf{x})
&=
Pr(|t|\leq z -  G^{-1}\hat q_2(z) , |\hat q_2(z)-q_2(z)|>\varepsilon_G\mid\mathbf{X}=\mathbf{x})\nonumber
\\&\quad+
Pr(|t|\leq z -  G^{-1}\hat q_2(z) , |\hat q_2(z)-q_2(z)|\leq\varepsilon_G\mid\mathbf{X}=\mathbf{x})\nonumber
\\
&\leq
Pr(|\hat q_2(z)-q_2(z)|>\varepsilon_G\mid\mathbf{X}=\mathbf{x})\nonumber
\\&\quad+
Pr(|t|\leq z -  G^{-1}q_2(z)+G^{-1}\varepsilon_G\mid\mathbf{X}=\mathbf{x})\nonumber
\\
&=
o(G^{-1})
+
Pr(|t|\leq z -  G^{-1}q_2(z)+G^{-1}\varepsilon_G\mid\mathbf{X}=\mathbf{x}).\label{eq:proof_Theorem1_eq1}
\end{align}
Similarly, we have
\begin{equation}
Pr(|t|\leq z -  G^{-1}\hat q_2(z) \mid\mathbf{X}=\mathbf{x})
\geq
o(G^{-1})
+
Pr(|t|\leq z -  G^{-1}q_2(z)-G^{-1}\varepsilon_G\mid\mathbf{X}=\mathbf{x}).\label{eq:proof_Theorem1_eq2}
\end{equation}
By the Cram\'er-Edgeworth expansion in \eqref{eq:edgeworth_expansion}, we have
\begin{align*}
&
Pr(|t|\leq z -  G^{-1}q_2(z)+G^{-1}\varepsilon_G\mid\mathbf{X}=\mathbf{x})
\\
&=
2 \Phi(z -  G^{-1}q_2(z)+G^{-1}\varepsilon_G) - 1 \\&\quad+2G^{-1}q_2(z -  G^{-1}q_2(z)+G^{-1}\varepsilon_G)\phi(z -  G^{-1}q_2(z)+G^{-1}\varepsilon_G)\\&\quad+o(G^{-1})
\\
&
Pr(|t|\leq z -  G^{-1}q_2(z)-G^{-1}\varepsilon_G\mid\mathbf{X}=\mathbf{x})
\\
&=
2 \Phi(z -  G^{-1}q_2(z)-G^{-1}\varepsilon_G) - 1 \\&\quad+2G^{-1}q_2(z -  G^{-1}q_2(z)-G^{-1}\varepsilon_G)\phi(z -  G^{-1}q_2(z)-G^{-1}\varepsilon_G)\\&\quad+o(G^{-1}).
\end{align*}
Since $\Phi$, $\phi$, and $q_2$ are continuously differentiable and $\varepsilon_G=o(1)$, we have
\begin{align*}
Pr(|t|\leq z -  G^{-1}q_2(z)+G^{-1}\varepsilon_G\mid\mathbf{X}=\mathbf{x})
&=
2 \Phi(z) - 1 +o(G^{-1})
\\
Pr(|t|\leq z -  G^{-1}q_2(z)-G^{-1}\varepsilon_G\mid\mathbf{X}=\mathbf{x})
&=
2 \Phi(z) - 1 +o(G^{-1}).
\end{align*}
Together with Eq.\eqref{eq:proof_Theorem1_eq1}-\eqref{eq:proof_Theorem1_eq2}, we have
$$
Pr(|t|\leq z -  G^{-1}\hat q_2(z) \mid\mathbf{X}=\mathbf{x})
-
(1-\alpha)
=
o(G^{-1}),
$$
which implies the statement of this theorem.
\end{proof}

\subsection{Proof of Theorem \ref{thm:theorem3}}\label{sec:proof_theorem3}
The proof structure is similar to \citet[Theorem 2]{bhattacharya1978validity}, but we extend it to the triangular array $\{\eta_g\}$ with potentially non-identical distributions.
Define $\mathbb{F}(z)=\Phi(z)+G^{-1/2} q_1(z)\phi(z)+G^{-1}q_2(z)\phi(z)$, where
\begin{align*}
q_1(z)&=-(k_1+\frac{1}{6}k_3He_2(z)),\\
q_2(z)&=-\left(\frac{1}{2}(k_2+k_1^2)He_1(z)+\frac{1}{24}(k_4+4k_1k_3)He_3(z)+\frac{1}{72}k_3^2He_5(z)\right).
\end{align*}
Let $\mathcal{N}$ be any bounded neighborhood of $\Phi^{-1}(1-\alpha/2)$.
In this proof, we are going to show
$$
\sup_{z\in\mathcal{N}}\left|Pr(t\leq z\mid\mathbf{X}=\mathbf{x})-\mathbb{F}(z)\right|=o(G^{-1})
$$
under the null. (We can show similarly that $\sup_{z\in\mathcal{N}}\left|Pr(t\geq -z\mid\mathbf{X}=\mathbf{x})-\mathbb{F}(-z)\right|=o(G^{-1})$ under the null.)
Since $(\omega_{1g},\omega_{2g}',\omega_{3g})'$ is a function of $\eta_g$, there is a function $\tilde{H}$ such that
$$
t
=
\tilde{H}\left(\mathbf{V}_G^{-1/2}\frac{1}{\sqrt{G}}\sum_{g=1}^G\eta_g;G^{-1/2},\gamma_1\right)
$$
with
\[
\mathbf V_G=
\frac{1}{G}\sum_{g=1}^G
Var(\eta_g\mid\mathbf{X}=\mathbf{x}),
\]
under the null hypothesis that
$$
\tilde{H}(u;s,\gamma_1)=u_1(1+s\tilde{h}_1(u)+s^2\tilde{h}_2(u))^{-1/2}
$$
for a linear function $\tilde{h}_1(u)$ and a quadratic function $\tilde{h}_2(u)$, where a finite-dimensional vector $\gamma_1$ collects the coefficients in $\tilde{h}_1(u)$ and $\tilde{h}_2(u)$.

In the following two lemmas, we show the Cram\'er condition and the Cram\'er-Edgeworth expansion for $\frac{1}{\sqrt{G}}\sum_{g=1}^G\eta_g$.

\begin{lem}\label{lem:cramer's_condition1}
$\limsup_{g\rightarrow\infty}\sup_{\|a\|\geq b}\left|E\left[\exp\left(ia'\eta_g\right)\mid\mathbf{X}=\mathbf{x}\right]\right|<1$ for every positive $b$.
\end{lem}
\begin{proof}
By  \citet[Theorem 1]{bobkov2012bounds}, there is a positive constant $c_1$ such that $\left|E\left[\exp\left(ia'\eta_g\right)\mid\mathbf{X}=\mathbf{x}\right]\right|\leq 1-c_1\min\{1/Var(a'\eta_g/\|a\|\mid\mathbf{X}=\mathbf{x}),\|a\|^2\}$. Let $c_2$ be a positive lower bound for the eigenvalues of $Var(\eta_g\mid\mathbf{X}=\mathbf{x})$. Then 
$\left|E\left[\exp\left(ia'\eta_g\right)\mid\mathbf{X}=\mathbf{x}\right]\right|\leq 1-c_1\min\{1/c_2,\|a\|^2\}$. Now we have 
$$
\limsup_{g\rightarrow\infty}\sup_{\|a\|\geq b}\left|E\left[\exp\left(ia'\eta_g\right)\mid\mathbf{X}=\mathbf{x}\right]\right|\leq 1-c_1\min\{1/c_2,b^2\}<1.
$$
\end{proof}

\begin{lem}\label{lem:BRtheorem20.6-1}
Under the null,
$$
\sup_{z\in\mathcal{N}}\left|Pr(t\leq z\mid\mathbf{X}=\mathbf{x})-\int_{\{\tilde{H}(u;G^{-1/2},\gamma_1)\leq z\}}\sum_{r=0}^2G^{-r/2}(\tilde{P}_r(-D;\gamma_2)\phi_I)(u)du\right|=o(G^{-1}),
$$
where $\phi_I$ is the multidimensional standard normal pdf, $\gamma_2$ is the collection of average cumulants up to fourth order, and $\tilde{P}_r$ is the polynomial defined in \citet[Section 7]{bhattacharya2010normal}.
\end{lem}
\begin{proof}
It follows from the Cram\'er condition (Lemma \ref{lem:cramer's_condition1}) and \citet[Theorem 20.6]{bhattacharya2010normal}.
\end{proof}

Define
$$
\xi(u;s,\gamma_2)=\sum_{r=0}^2s^r(\tilde{P}_r(-D;\gamma_2)\phi_I)(u).
$$
Then Lemma \ref{lem:BRtheorem20.6-1} is
$$
\sup_{z\in\mathcal{N}}\left|Pr(t\leq z\mid\mathbf{X}=\mathbf{x})-\int_{\{\tilde{H}(u;G^{-1/2},\gamma_1)\leq z\}}\xi(u;G^{-1/2},\gamma_2)du\right|=o(G^{-1})
$$
under the null. The following lemma shows the integral can be expressed with polynomials of $z$.

\begin{lem}\label{lem:BGlemma2.1}
There are finite-order polynomials $\tilde{q}_1(z)$ and $\tilde{q}_2(z)$ such that
\begin{equation}\label{eq:BG_lemma2.1_eq1}
\sup_{z\in\mathbb{R}}\left|\int_{\{\tilde{H}(u;G^{-1/2},\gamma_1)\leq z\}}\xi(u;G^{-1/2},\gamma_2)du-\tilde{\mathbb{F}}(z)\right|=o(G^{-1}),
\end{equation}
and that
\begin{equation}\label{eq:approx_moment_F_tilde}
\int
\tilde{H}(u;G^{-1/2},\gamma_1)^j\xi(u;G^{-1/2},\gamma_2)(u)du
=
\int z^jd\tilde{\mathbb{F}}(z)+o(G^{-1})
\end{equation}
for every $j$, where $\tilde{\mathbb{F}}(z)=\Phi(z)+G^{-1/2} \tilde{q}_1(z)\phi(z)+G^{-1}\tilde{q}_2(z)\phi(z)$.\footnote{The coefficients of the polynomials $\tilde{q}_1(z)$ and $\tilde{q}_2(z)$ can depend on $G$ because the average cumulants can depend on $G$.}
\end{lem}
\begin{proof}
The proof is done by closely following  the one in \citet[Lemma 2.1]{bhattacharya1978validity}. 
\end{proof}

For the rest of the proof, we are going to show $\sup_{z\in\mathcal{N}}\left|\tilde{\mathbb{F}}(z)-\mathbb{F}(z)\right|=o(G^{-1})$.
Although Lemma \ref{lem:BGlemma2.1} does not specify the orders of the polynomials $\tilde{q}_1(z)$ and $\tilde{q}_2(z)$, the orders are finite. As a result, it suffices to show all the moments are different only in the magnitude of $o(G^{-1})$, i.e., $\int z^jd\mathbb{F}(z)=\int z^jd\tilde{\mathbb{F}}(z)+o(G^{-1})$. Define
$$
\tilde{t}
=
\sqrt{G}W_{1}\left(1+\frac{1}{2}W_{2}'\Gamma W_{2}-\frac{1}{2}W_{3}+\frac{3}{8}W_{3}^2\right),
$$
where
$$
\left(\begin{array}{c}W_1\\ W_2\\W_3\end{array}\right)
=
\frac{1}{G}\sum_{g=1}^G\left(\begin{array}{c}\omega_{1g}\\\omega_{2g}\\ \omega_{3g}\end{array}\right).
$$
Since $(\omega_{1g},\omega_{2g}',\omega_{3g})'$ is a function of $\eta_g$, there is a polynomial $H$ such that
$$
\tilde{t}
=
H\left(\mathbf{V}_G^{-1/2}\frac{1}{\sqrt{G}}\sum_{g=1}^G\eta_g;G^{-1/2},\gamma_3\right)
$$
and that $H(u;s,\gamma_3)=u_1(1+sh_1(u)+s^2h_2(u))$ for a linear function $h_1(u)$ and a quadratic function $h_2(u)$, where a finite-dimensional vector $\gamma_3$ collects the coefficients in $h_1(u)$ and $h_2(u)$. For the rest of the proof, we will show each equality of the following expression:
\begin{align*}
\int z^jd\mathbb{F}(z)
&=
E[\tilde{t}^j\mid\mathbf{X}=\mathbf{x}]+o(G^{-1})
\\
&=
\int H(u;G^{-1/2},\gamma_3)^j\xi(u;G^{-1/2},\gamma_2)du+o(G^{-1})
\\
&=
\int \tilde{H}(u;G^{-1/2},\gamma_1)^j\xi(u;G^{-1/2},\gamma_2)du+o(G^{-1})
\\
&=
\int z^jd\tilde{\mathbb{F}}(z)+o(G^{-1}),
\end{align*}
where the first three equalities will be shown in Lemma \ref{lem:F_moments_t_tilde_moments},  \ref{lem:bound_derivative_character}, and \ref{lem:lemma_t_expansion_tilde_t}, and the last equality has been shown in Eq. \eqref{eq:approx_moment_F_tilde}.

We first provide a few lemmas (Lemma \ref{lem:mu_omega_3_disappear}-\ref{lem:cumulants_for_tilde_t}) to characterize the first four moments of $\tilde{t}$, and then we show Lemma \ref{lem:F_moments_t_tilde_moments}-\ref{lem:lemma_t_expansion_tilde_t}.

\begin{lem}\label{lem:mu_omega_3_disappear}
Define
\begin{align*}
\mu_{1,3}
&=
\frac{1}{G}\sum_{g}E\left[
\omega_{1g}\omega_{3g}\mid\mathbf{X}=\mathbf{x}
\right]
\\
\mu_{3,3}
&=
\frac{1}{G}\sum_{g}E\left[
\omega_{3g}^2\mid\mathbf{X}=\mathbf{x}
\right],
\\
\mu_{1,1,3}
&=
\frac{1}{G}\sum_{g}
E\left[
\omega_{1g}^2
\omega_{3g}\mid\mathbf{X}=\mathbf{x}
\right],
\end{align*}
and
$$
\theta
=
\frac{1}{G}\sum_{g=1}^G(\sigma_{g}^2/\sigma^2)^2
=
\frac{1}{G}\sum_{g=1}^G
\left(\sigma^{-2}\lambda'\Pi X_g'E[u_gu_g'\mid\mathbf{X}=\mathbf{x}]X_g\Pi\lambda
\right)^2.
$$
Then $\mu_{1,3}=\mu_{1,1,1}$, $\mu_{3,3}=\mu_{1,1,1,1}-\theta$, and $\mu_{1,1,3}=\mu_{1,1,1,1}-\theta$.
\end{lem}
\begin{proof}
Since $\omega_{1g}=\sigma^{-1}\lambda'\Pi X_g'u_g$ and $\omega_{3g}=\omega_{1g}^2-E[\omega_{1g}^2\mid\mathbf{X}=\mathbf{x}]$,
we have
\begin{align*}
\mu_{1,3}
&=
\frac{1}{G}\sum_{g}E\left[
\omega_{1g}(\omega_{1g}^2-E[\omega_{1g}^2\mid\mathbf{X}=\mathbf{x}])\mid\mathbf{X}=\mathbf{x}
\right]
\\
&=
\frac{1}{G}\sum_{g}E\left[
\omega_{1g}^3\mid\mathbf{X}=\mathbf{x}\right]
-
\frac{1}{G}\sum_{g}E\left[\omega_{1g}\mid\mathbf{X}=\mathbf{x}\right]
E[\omega_{1g}^2\mid\mathbf{X}=\mathbf{x}]
\\
&=
\frac{1}{G}\sum_{g}E\left[
\omega_{1g}^3\mid\mathbf{X}=\mathbf{x}\right]
\\
&=
\mu_{1,1,1}
\\
\mu_{1,1,3}
&=
\frac{1}{G}\sum_{g}
E\left[
\omega_{1g}^2
(\omega_{1g}^2-E[\omega_{1g}^2\mid\mathbf{X}=\mathbf{x}])\mid\mathbf{X}=\mathbf{x}
\right]
\\
&=
\frac{1}{G}\sum_{g}
E\left[
\omega_{1g}^4\mid\mathbf{X}=\mathbf{x}\right]
-
\frac{1}{G}\sum_{g}
E[\omega_{1g}^2\mid\mathbf{X}=\mathbf{x}]^2
\\
&=
\mu_{1,1,1,1}-\theta
\\
\mu_{3,3}
&=
\frac{1}{G}\sum_{g}E\left[
(\omega_{1g}^2-E[\omega_{1g}^2\mid\mathbf{X}=\mathbf{x}])^2\mid\mathbf{X}=\mathbf{x}
\right],
\\
&=
\frac{1}{G}\sum_{g}
E\left[
\omega_{1g}^4\mid\mathbf{X}=\mathbf{x}\right]
-
\frac{1}{G}\sum_{g}
E[\omega_{1g}^2\mid\mathbf{X}=\mathbf{x}]^2
\\
&=
\mu_{1,1,1,1}-\theta.
\end{align*}
\end{proof}

\begin{lem}\label{lem:A.3DNM} (Lemma A.3 of \cite{djogbenou2019asymptotic})
Let $L$ be an integer in $\{2,\ldots,8\}$. For every $g=1,2,...$, suppose we have $L$ random variables $\check{\omega}_{1g},\ldots,\check{\omega}_{Lg}$ with $E\left[\check{\omega}_{\ell g}\mid\mathbf{X}=\mathbf{x}\right]=0$. If $(\check{\omega}_{1g},\ldots,\check{\omega}_{Lg})$'s are independent across $g$ and $E\left[\check{\omega}_{\ell g}^L\mid\mathbf{X}=\mathbf{x}\right]$ is bounded uniformly in $g$, then
$$
E\left[\left(\frac{1}{G}\sum_{g=1}^G\check{\omega}_{1g}\right)\cdots\left(\frac{1}{G}\sum_{g=1}^G\check{\omega}_{Lg}\right)\mid\mathbf{X}=\mathbf{x}\right]
=
\begin{cases}
O(G^{-L/2})&\mbox{ if $L$ is even }\\
O(G^{-(L+1)/2})&\mbox{ otherwise}.
\end{cases}
$$
\end{lem}

\begin{lem}\label{lem:cumulants_for_tilde_t}
\begin{align*}
E[\tilde{t}\mid\mathbf{X}=\mathbf{x}]
&=
G^{-1/2}\nu_1+o(G^{-1}),
\\
E[\tilde{t}^2\mid\mathbf{X}=\mathbf{x}]
&=
1+G^{-1}\nu_2+o(G^{-1}),
\\
E[\tilde{t}^3\mid\mathbf{X}=\mathbf{x}]
&=
G^{-1/2}\nu_3+o(G^{-1}),\mbox{ and }
\\
E[\tilde{t}^4\mid\mathbf{X}=\mathbf{x}]
&=
3+G^{-1}\nu_4+o(G^{-1}).
\end{align*}
\end{lem}
\begin{proof}
First, we are going to show the first moment.
Note that
\begin{align*}
G^{1/2}E\left[W_1\mid\mathbf{X}=\mathbf{x}\right]
&=
0
\end{align*}
and that
\begin{align*}
G^{1/2}E\left[W_1W_3\mid\mathbf{X}=\mathbf{x}\right]
&=
G^{-3/2}\sum_{g_1,g_2}
E\left[
\omega_{1g_1}
\omega_{3g_2}\mid\mathbf{X}=\mathbf{x}
\right]
\\
&=
G^{-1/2}\frac{1}{G}\sum_{g}
E\left[
\omega_{1g}
\omega_{3g}\mid\mathbf{X}=\mathbf{x}
\right]
\\
&=
G^{-1/2}\mu_{1,3},
\end{align*}
where the second equality holds because $\omega_{1g}$ and $\omega_{3g}$ are mean zero. By Lemma \ref{lem:A.3DNM},
$$
G^{1/2}E\left[W_1W_{j_1}'W_{j_2}\mid\mathbf{X}=\mathbf{x}
\right]
=
o(G^{-1})
$$
for $j_1,j_2=1,2,3$.
Since
$$
\tilde{t}=
\sqrt{G}W_1\left(1+\frac{1}{2}W_{2}'\Gamma W_{2}-\frac{1}{2}W_3+\frac{3}{8}W_3^2\right),
$$
we have
$$
E[\tilde{t}\mid\mathbf{X}=\mathbf{x}]=-G^{-1/2}\frac{\mu_{1,3}}{2}+o(G^{-1}).
$$
By Lemma \ref{lem:mu_omega_3_disappear}, we have the first equation of this lemma.

Second, we are going to show the second moment. Note that $GE\left[W_1^2\mid\mathbf{X}=\mathbf{x}\right]=1$, and that
\begin{align*}
GE\left[W_1^2W_3\mid\mathbf{X}=\mathbf{x}\right]
&=
G^{-2}\sum_{g_1,g_2,g_3}
E\left[
\omega_{1g_1}
\omega_{1g_2}
\omega_{3g_3}\mid\mathbf{X}=\mathbf{x}
\right]
\\
&=
G^{-1}\frac{1}{G}\sum_{g}
E\left[
\omega_{1g}^2
\omega_{3g}\mid\mathbf{X}=\mathbf{x}
\right]
\\
&=
G^{-1}\mu_{1,1,3},
\end{align*}
and that
\begin{align*}
&
GE\left[
W_1^2
W_{j_1}'W_{j_2}\mid\mathbf{X}=\mathbf{x}
\right]
\\
&=
G^{-3}\sum_{g_1,g_2,g_3,g_4}
E\left[
\omega_{1g_1}
\omega_{1g_2}
\omega_{j_1g_3}'
\omega_{j_2g_4}\mid\mathbf{X}=\mathbf{x}
\right]
\\
&=
G^{-3}\sum_{g_1}
E\left[
\omega_{1g_1}^2
\omega_{j_1g_1}'
\omega_{j_2g_1}\mid\mathbf{X}=\mathbf{x}
\right]
\\
&+
G^{-3}\sum_{g_1\ne g_2}
\left(E\left[
\omega_{1g_2}^2\mid\mathbf{X}=\mathbf{x}
\right]
E\left[
\omega_{j_1g_1}'
\omega_{j_2g_1}\mid\mathbf{X}=\mathbf{x}\right]
+2E\left[
\omega_{1g_1}\omega_{j_1g_1}'\mid\mathbf{X}=\mathbf{x}
\right]E\left[
\omega_{1g_2}\omega_{j_2g_2}\mid\mathbf{X}=\mathbf{x}
\right]
\right)
\\
&=
G^{-1}\left(\frac{1}{G}\sum_{g_1}E\left[
\omega_{j_1g_1}'
\omega_{j_2g_1}\mid\mathbf{X}=\mathbf{x}\right]
+2\frac{1}{G}\sum_{g_1}E\left[
\omega_{1g_1}\omega_{j_1g_1}'\mid\mathbf{X}=\mathbf{x}
\right]\frac{1}{G}\sum_{g_2}E\left[
\omega_{1g_2}
\omega_{j_2g_2}\mid\mathbf{X}=\mathbf{x}
\right]\right)
\\&\quad+
o(G^{-1})
\\
&=
G^{-1}\left(\mu_{j_1,j_2}
+2\mu_{1,j_1}'\mu_{1,j_2}\right)
+
o(G^{-1})
\end{align*}
for $j_1,j_2=1,2,3$.
By Lemma \ref{lem:A.3DNM},
$$
GE\left[
W_1^2
W_{j_1}
\cdots
W_{j_k}\mid\mathbf{X}=\mathbf{x}
\right]
=
o(G^{-1})
\mbox{ for }k\geq 3
$$
where $W_{j_1},\ldots,W_{j_k}$ are elements of $(W_1,W_2',W_3)'$.
Since
\begin{align*}
\tilde{t}^2
&=
GW_1^2\left(1+\frac{1}{2}W_{2}'\Gamma W_{2}-\frac{1}{2}W_3+\frac{3}{8}W_3^2\right)^2
\\
&=
GW_1^2\left(1-W_3+W_3^2+W_{2}'\Gamma W_{2}+(\mbox{products of three or four terms of }W_2,W_3)\right),
\end{align*}
we have
$$
E[\tilde{t}^2\mid\mathbf{X}=\mathbf{x}]=1+G^{-1}\left(-\mu_{1,1,3}+(\mu_{3,3}
+2\mu_{1,3}^2)
+
(\mu_{2,2}
+2\mu_{1,2}'\Gamma\mu_{1,2})
\right)+o(G^{-1}).
$$
By Lemma \ref{lem:mu_omega_3_disappear}, we have the second equation of this lemma.

Third, we are going to show the third moment. Note that
\begin{align*}
G^{3/2}E\left[
W_1^3\mid\mathbf{X}=\mathbf{x}
\right]
&=
G^{-3/2}
\sum_{g_1,g_2,g_3}E\left[
\omega_{1g_1}
\omega_{1g_2}
\omega_{1g_3}\mid\mathbf{X}=\mathbf{x}
\right]
\\
&=
G^{-3/2}
\sum_{g_1}E\left[
\omega_{1g_1}^3\mid\mathbf{X}=\mathbf{x}
\right]
\\
&=
G^{-1/2}\mu_{1,1,1},
\end{align*}
and that
\begin{align*}
&
G^{3/2}E\left[
W_1^3W_3\mid\mathbf{X}=\mathbf{x}
\right]
\\
&=
G^{-5/2}\sum_{g_1,g_2,g_3,g_4}
E\left[
\omega_{1g_1}
\omega_{1g_2}
\omega_{1g_3}
\omega_{3g_4}\mid\mathbf{X}=\mathbf{x}
\right]
\\
&=
G^{-5/2}\sum_{g_1}
E\left[
\omega_{1g_1}^3
\omega_{3g_1}\mid\mathbf{X}=\mathbf{x}
\right]
\\
&+
G^{-5/2}\sum_{g_1\ne g_2}
\left(E\left[
\omega_{1g_2}^2\mid\mathbf{X}=\mathbf{x}
\right]
E\left[
\omega_{1g_1}
\omega_{3g_1}\mid\mathbf{X}=\mathbf{x}\right]
+2E\left[
\omega_{1g_1}^2\mid\mathbf{X}=\mathbf{x}
\right]E\left[
\omega_{1g_2}\omega_{3g_2}\mid\mathbf{X}=\mathbf{x}
\right]
\right)
\\
&=
3G^{-1/2}
\frac{1}{G}\sum_{g_1}
E\left[
\omega_{1g_1}
\omega_{3g_1}\mid\mathbf{X}=\mathbf{x}\right]
+
o(G^{-1})
\\
&=
3G^{-1/2}\mu_{1,3}
+
o(G^{-1}).
\end{align*}
By Lemma \ref{lem:A.3DNM},
$$
G^{3/2}E\left[
W_1^3
W_{j_1}
\cdots
W_{j_k}\mid\mathbf{X}=\mathbf{x}
\right]
=
o(G^{-1})\mbox{ for }k\geq 2.
$$
Since
\begin{align*}
\tilde{t}^3
&=
G^{3/2}W_1^3\left(1+\frac{1}{2}W_{2}'\Gamma W_{2}-\frac{1}{2}W_3+\frac{3}{8}W_3^2\right)^3
\\
&=
G^{3/2}W_1^3\left(1-\frac{3}{2}W_3+(\mbox{products of two to six terms of }W_2,W_3)\right),
\end{align*}
we have
$$
E[\tilde{t}^3\mid\mathbf{X}=\mathbf{x}]=G^{-1/2}(\mu_{1,1,1}-\frac{9}{2}\mu_{1,3})+o(G^{-1}).
$$
By Lemma \ref{lem:mu_omega_3_disappear}, we have the third equation of this lemma.

Fourth, we are going to show the fourth moment. Note that
\begin{align*}
&
G^2E\left[
W_1^4\mid\mathbf{X}=\mathbf{x}
\right]
\\
&=
G^{-2}\sum_{g_1,g_2,g_3,g_4}
E\left[
\omega_{1g_1}
\omega_{1g_2}
\omega_{1g_3}
\omega_{1g_4}\mid\mathbf{X}=\mathbf{x}
\right]
\\
&=
G^{-2}\sum_{g_1}
E\left[
\omega_{1g_1}^4\mid\mathbf{X}=\mathbf{x}
\right]
+
3\left(\frac{1}{G}\sum_{g_1}
E\left[
\omega_{1g_1}^2\mid\mathbf{X}=\mathbf{x}
\right]
\right)
\left(
G^{-1}
\sum_{g_2\ne g_1}
E\left[
\omega_{1g_2}^2\mid\mathbf{X}=\mathbf{x}
\right]
\right)
\\
&=
3+
G^{-1}\frac{1}{G}\sum_{g_1}\left(
E\left[
\omega_{1g_1}^4\mid\mathbf{X}=\mathbf{x}
\right]
-
3
E\left[
\omega_{1g_1}^2\mid\mathbf{X}=\mathbf{x}
\right]^2\right)
\\
&=
3+G^{-1}(\mu_{1,1,1,1}-3\theta),
\end{align*}
that
\begin{align*}
&
G^2E\left[
W_1^4W_3\mid\mathbf{X}=\mathbf{x}
\right]
\\
&=
G^{-3}\sum_{g_1,g_2,g_3,g_4,g_5}
E\left[
\omega_{3g_1}
\omega_{1g_2}
\omega_{1g_3}
\omega_{1g_4}
\omega_{1g_5}\mid\mathbf{X}=\mathbf{x}
\right]
\\
&=
G^{-3}\sum_{g_1,g_2}
\left(
6E\left[
\omega_{3g_1}
\omega_{1g_1}^2\mid\mathbf{X}=\mathbf{x}
\right]
E\left[
\omega_{1g_2}^2\mid\mathbf{X}=\mathbf{x}
\right]
+4
E\left[
\omega_{3g_1}
\omega_{1g_1}\mid\mathbf{X}=\mathbf{x}
\right]
E\left[
\omega_{1g_2}^3\mid\mathbf{X}=\mathbf{x}
\right]
\right)
\\&+G^{-3}\sum_{g_1}
E\left[
\omega_{3g_1}
\omega_{1g_1}
\omega_{1g_1}
\omega_{1g_1}
\omega_{1g_1}\mid\mathbf{X}=\mathbf{x}
\right]
+o(G^{-1})
\\
&=
G^{-1}
\left(6\mu_{1,1,3}
+
4\mu_{1,3}
\mu_{1,1,1}\right)
+o(G^{-1}),
\end{align*}
and that
\begin{align*}
&
G^2E\left[
W_1^4W_{j_1}'W_{j_2}\mid\mathbf{X}=\mathbf{x}
\right]
\\
&=
G^{-4}\sum_{g_1,g_2,g_3,g_4,g_5,g_6}
E\left[
\omega_{j_1g_1}'
\omega_{j_2g_2}
\omega_{1g_3}
\omega_{1g_4}
\omega_{1g_5}
\omega_{1g_6}\mid\mathbf{X}=\mathbf{x}
\right]
\\
&=
G^{-4}\sum_{g_1,g_2,g_3}
\left(
3E\left[
\omega_{j_1g_1}'
\omega_{j_2g_1}
\omega_{1g_2}^2
\omega_{1g_3}^2\mid\mathbf{X}=\mathbf{x}
\right]
+
12E\left[
\omega_{j_1g_1}'
\omega_{j_2g_2}
\omega_{1g_1}
\omega_{1g_2}
\omega_{1g_3}^2\mid\mathbf{X}=\mathbf{x}
\right]
\right)
\\&+
G^{-4}\sum_{g_1,g_2}
\left(4E\left[
\omega_{j_1g_1}'
\omega_{j_2g_1}
\omega_{1g_1}
\omega_{1g_2}^3\mid\mathbf{X}=\mathbf{x}
\right]
+6
E\left[
\omega_{j_1g_1}'
\omega_{j_2g_2}
\omega_{1g_1}^2
\omega_{1g_2}^2\mid\mathbf{X}=\mathbf{x}
\right]
\right)
\\&+
G^{-4}\sum_{g_1}
E\left[
\omega_{j_1g_1}'
\omega_{j_2g_1}
\omega_{1g_1}
\omega_{1g_1}
\omega_{1g_1}
\omega_{1g_1}\mid\mathbf{X}=\mathbf{x}
\right]
+o(G^{-1})
\\
&=
G^{-4}\sum_{g_1,g_2,g_3}
\left(
3E\left[
\omega_{j_1g_1}'
\omega_{j_2g_1}
\omega_{1g_2}^2
\omega_{1g_3}^2\mid\mathbf{X}=\mathbf{x}
\right]
+
12E\left[
\omega_{j_1g_1}'
\omega_{j_2g_2}
\omega_{1g_1}
\omega_{1g_2}
\omega_{1g_3}^2\mid\mathbf{X}=\mathbf{x}
\right]
\right)
+o(G^{-1})
\\
&=
G^{-1}
\left(
3\mu_{j_1,j_2}
+
12\mu_{1,j_1}'\mu_{1,j_2}
\right)
+o(G^{-1}).
\end{align*}
By Lemma \ref{lem:A.3DNM},
$$
G^{2}
E\left[
W_1^4
W_{j_1}
\cdots
W_{j_k}\mid\mathbf{X}=\mathbf{x}
\right]
=
o(G^{-1})
\mbox{ for }k\geq 3.
$$
Since
\begin{align*}
\tilde{t}^4
&=
G^2W_1^4\left(1+\frac{1}{2}W_{2}'\Gamma W_{2}-\frac{1}{2}W_3+\frac{3}{8}W_3^2\right)^4
\\
&=
G^2W_1^4\left(1-2W_3+2W_{2}'\Gamma W_{2}+3W_3^2+(\mbox{products of three to eight terms of }W_2,W_3)\right),
\end{align*}
we have
\begin{align*}
&E[\tilde{t}^4\mid\mathbf{X}=\mathbf{x}]-3
\\&=G^{-1}((\mu_{1,1,1,1}-3\theta)-2(6\mu_{1,1,3}
+4
\mu_{1,3}
\mu_{1,1,1})+2\left(
3\mu_{2,2}
+
12\mu_{1,2}'\Gamma\mu_{1,2}
\right)+3\left(
3\mu_{3,3}
+
12\mu_{1,3}'\mu_{1,3}
\right))\\&\quad+o(G^{-1}).
\end{align*}
By Lemma \ref{lem:mu_omega_3_disappear}, we have the fourth equation of this lemma.
\end{proof}

\begin{lem}\label{lem:F_moments_t_tilde_moments}
$\int z^jd\mathbb{F}(z)=E[\tilde{t}^j\mid\mathbf{X}=\mathbf{x}]+o(G^{-1})$ for every positive integer $j$.
\end{lem}
\begin{proof}
By Lemma \ref{lem:cumulants_for_tilde_t} and the definitions of $q_1$ and $q_2$, the $j$-th cumulant of $\mathbb{F}$ differs only in magnitude of $o(G^{-1})$ from that of $\tilde{t}$ for $j\leq 4$.
Note that, by the definitions of $q_1$ and $q_2$, the $j$th cumulant of $\mathbb{F}$ is $o(G^{-1})$ for $j\geq 5$. Since all the moments are bounded (Assumption \ref{assumption:moments_and_invertibility} (iii)), by \citet[Theorem 2.1]{hall1992bootstrap}, the $j$th cumulant of $\tilde{t}$ is $o(G^{-1})$  for $j\geq 5$.
\end{proof}

\begin{lem}\label{lem:bound_derivative_character}
$E[\tilde{t}^j\mid\mathbf{X}=\mathbf{x}]-\int H(u;G^{-1/2},\gamma_3)^j\xi(u;G^{-1/2},\gamma_2)du=o(G^{-1})$ for every positive integer $j$.
\end{lem}
\begin{proof}
Since $H(u;G^{-1/2},\gamma_3)^j$ is a polynomial of $u$, its moments are expressed as the derivatives of the characteristic function at the origin. \citet[Theorem 9.9]{bhattacharya2010normal} implies the statement of this lemma.
\end{proof}

\begin{lem}\label{lem:lemma_t_expansion_tilde_t}
$\int (\tilde{H}(u;G^{-1/2},\gamma_1)^j-H(u;G^{-1/2},\gamma_3)^j)\xi(u;G^{-1/2},\gamma_2)du=o(G^{-1})$ for every $j$.
\end{lem}
\begin{proof}
Under $H_0$, we have
$$
t-\tilde{t}=-\sqrt{G}W_{1}\left(
1+\frac{1}{2}W_{2}'\Gamma W_{2}-\frac{1}{2}W_{3}+\frac{3}{8}W_{3}^2-(1-W_{2}'\Gamma W_{2}+W_3)^{-1/2}
\right).
$$
By the second-order Taylor expansion of $u\mapsto (1+u)^{-1/2}$, we have
\begin{align*}
&|(1-W_{2}'\Gamma W_{2}+W_3)^{-1/2}
-(1-\frac{1}{2}(-W_{2}'\Gamma W_{2}+W_3)+\frac{3}{8}(-W_{2}'\Gamma W_{2}+W_3)^2)|\\&\leq \frac{1}{3!}\left(\frac{d^3}{du^3}(1+u)^{-1/2}|_{u=\bar{u}}\right)|-W_{2}'\Gamma W_{2}+W_3|^3\\&=\frac{15}{48}(1+\bar{u})^{-7/2}|-W_{2}'\Gamma W_{2}+W_3|^3
\end{align*}
for some point $\bar{u}$ between $0$ and $-W_{2}'\Gamma W_{2}+W_3$.
Therefore,
$$
|t-\tilde{t}|\leq\sqrt{G}|W_{1}|\left(
\frac{3}{8}|(W_{2}'\Gamma W_{2})^2-2(W_{2}'\Gamma W_{2})W_3|+
|-W_{2}'\Gamma W_{2}+W_3|^3
\right)
$$
when $|-W_{2}'\Gamma W_{2}+W_3|$ is sufficiently small.
Therefore, we have
$$
|\tilde{H}(u;G^{-1/2},\gamma_1)-H(u;G^{-1/2},\gamma_3)|
\leq G^{-3/2}\sum_{r=4}^7G^{-(r-4)/2}|\tilde{g}_r(u)|
$$
where $\tilde{g}_r(u)$ is an $r$th order polynomial.
As a result, we have
\begin{align*}
&\int (\tilde{H}(u;G^{-1/2},\gamma_1)-H(u;G^{-1/2},\gamma_3))^j\xi(u;G^{-1/2},\gamma_2)du
\\
&\leq
\int(G^{-3/2}\sum_{r=4}^7G^{-(r-4)/2}|\tilde{g}_r(u)|)^j\xi(u;G^{-1/2},\gamma_2)du
\\
&=
O(G^{-3/2}).
\end{align*}
Since $\int \tilde{H}(u;G^{-1/2},\gamma_1)^j\xi(u;G^{-1/2},\gamma_2)du=O(1)$ for every $j$, the statement of this lemma follows.
\end{proof}

\newpage
\section{Additional Tables}\label{sec:addendix_sim_tables}

\begin{table}[!htbp]
\centering
\caption{Monte Carlo rejection probabilities in Design 1 (BDM)}
\label{tab:bdm-en-21-g-10-200-alpha-0p05-nmc-10000-nboot-1000}
 
\begin{tabular}{@{}rrrrrr@{}}
\toprule
\multicolumn{1}{c}{$G$} & \multicolumn{1}{c}{Normal} & \multicolumn{1}{c}{Student $d_1$} & \multicolumn{1}{c}{Pairs} & \multicolumn{1}{c}{WCB} & \multicolumn{1}{c}{Analytical} \\
\midrule
10  & 0.114 & 0.051 & 0.009 & 0.042 & 0.046 \\
25  & 0.074 & 0.052 & 0.042 & 0.053 & 0.049 \\
50  & 0.061 & 0.051 & 0.050 & 0.053 & 0.050 \\
75  & 0.056 & 0.049 & 0.050 & 0.051 & 0.048 \\
100 & 0.053 & 0.048 & 0.048 & 0.049 & 0.047 \\
200 & 0.051 & 0.050 & 0.051 & 0.051 & 0.050 \\
\bottomrule
\end{tabular}
\end{table}
 \begin{table}[!htbp]
\centering
\caption{Actual critical values in Design 1 (BDM)}
\label{tab:bdm-en-21-g-10-200-alpha-0p05-nmc-10000-nboot-1000-critical-values}
\begingroup
 
\setlength{\tabcolsep}{3.5pt}
\begin{tabular}{@{}rrrrrrrrr@{}}
\toprule
\multicolumn{1}{c}{$G$} & \multicolumn{1}{c}{Student $d_1$} & \multicolumn{1}{c}{Student $d_2$} & \multicolumn{1}{c}{Pairs} & \multicolumn{1}{c}{WCB} & \multicolumn{1}{c}{Analytical} & \multicolumn{1}{c}{Cornish-Fisher} & \multicolumn{1}{c}{Actual} \\
\midrule
10  & 2.508 & 2.385 & 3.580 & 2.566 & 2.510 & 2.374 & 2.517 \\
25  & 2.148 & 2.106 & 2.232 & 2.142 & 2.155 & 2.133 & 2.165 \\
50  & 2.050 & 2.030 & 2.066 & 2.042 & 2.053 & 2.047 & 2.064 \\
75  & 2.019 & 2.006 & 2.023 & 2.012 & 2.021 & 2.015 & 2.009 \\
100 & 2.004 & 1.994 & 2.003 & 1.997 & 2.005 & 2.003 & 1.981 \\
200 & 1.982 & 1.977 & 1.978 & 1.975 & 1.982 & 1.981 & 1.977 \\
\bottomrule
\end{tabular}
\endgroup
\end{table}
 \newpage
\begin{table}[!htbp]
\centering
\caption{Monte Carlo rejection probabilities in Design 2 (exponential)}
\label{tab:dmn-en-1-g-10-200-alpha-0p05-nmc-10000-nboot-1000}
 
\begin{tabular}{@{}rrrrrr@{}}
\toprule
\multicolumn{1}{c}{$G$} & \multicolumn{1}{c}{Normal} & \multicolumn{1}{c}{Student $d_1$} & \multicolumn{1}{c}{Pairs} & \multicolumn{1}{c}{WCB} & \multicolumn{1}{c}{Analytical} \\
\midrule
10  & 0.140 & 0.098 & 0.068 & 0.094 & 0.089 \\
25  & 0.097 & 0.078 & 0.061 & 0.079 & 0.066 \\
50  & 0.072 & 0.064 & 0.054 & 0.066 & 0.055 \\
75  & 0.069 & 0.065 & 0.056 & 0.067 & 0.056 \\
100 & 0.064 & 0.060 & 0.056 & 0.062 & 0.054 \\
200 & 0.056 & 0.055 & 0.051 & 0.056 & 0.050 \\
\bottomrule
\end{tabular}
\end{table}
 \begin{table}[!htbp]
\centering
\caption{Actual critical values in Design 2 (exponential)}
\label{tab:dmn-en-1-g-10-200-alpha-0p05-nmc-10000-nboot-1000-critical-values}
\begingroup
 
\setlength{\tabcolsep}{3.5pt}
\begin{tabular}{@{}rrrrrrrrr@{}}
\toprule
\multicolumn{1}{c}{$G$} & \multicolumn{1}{c}{Student $d_1$} & \multicolumn{1}{c}{Student $d_2$} & \multicolumn{1}{c}{Pairs} & \multicolumn{1}{c}{WCB} & \multicolumn{1}{c}{Analytical} & \multicolumn{1}{c}{Cornish-Fisher} & \multicolumn{1}{c}{Actual} \\
\midrule
10  & 2.385 & 2.385 & 2.848 & 2.341 & 2.479 & 3.097 & 3.153 \\
25  & 2.106 & 2.106 & 2.315 & 2.070 & 2.234 & 2.411 & 2.426 \\
50  & 2.030 & 2.030 & 2.147 & 2.006 & 2.121 & 2.180 & 2.192 \\
75  & 2.006 & 2.006 & 2.090 & 1.987 & 2.076 & 2.108 & 2.144 \\
100 & 1.994 & 1.994 & 2.057 & 1.979 & 2.050 & 2.068 & 2.091 \\
200 & 1.977 & 1.977 & 2.008 & 1.967 & 2.008 & 2.014 & 2.017 \\
\bottomrule
\end{tabular}
\endgroup
\end{table}
 \newpage
\begin{table}[!htbp]
\centering
\caption{Monte Carlo rejection probabilities in Design 3 (binary regressor)}
\label{tab:generic-en-1-g-10-200-alpha-0p05-nmc-10000-nboot-500}
 
\begin{tabular}{@{}rrrrrr@{}}
\toprule
\multicolumn{1}{c}{$G$} & \multicolumn{1}{c}{Normal} & \multicolumn{1}{c}{Student $d_1$} & \multicolumn{1}{c}{Pairs} & \multicolumn{1}{c}{WCB} & \multicolumn{1}{c}{Analytical} \\
\midrule
10  & 0.172 & 0.110 & 0.042 & 0.097 & 0.104 \\
25  & 0.099 & 0.079 & 0.058 & 0.079 & 0.068 \\
50  & 0.073 & 0.064 & 0.054 & 0.065 & 0.056 \\
75  & 0.067 & 0.060 & 0.054 & 0.063 & 0.054 \\
100 & 0.067 & 0.062 & 0.055 & 0.063 & 0.055 \\
200 & 0.056 & 0.054 & 0.051 & 0.056 & 0.049 \\
\bottomrule
\end{tabular} 
\end{table}
 \begin{table}[!htbp]
\centering
\caption{Actual critical values in Design 3 (binary regressor)}
\label{tab:generic-en-1-g-10-200-alpha-0p05-nmc-10000-nboot-1000-critical-values}
\begingroup
 
\setlength{\tabcolsep}{3.5pt}
\begin{tabular}{@{}rrrrrrrrr@{}}
\toprule
\multicolumn{1}{c}{$G$} & \multicolumn{1}{c}{Student $d_1$} & \multicolumn{1}{c}{Student $d_2$} & \multicolumn{1}{c}{Pairs} & \multicolumn{1}{c}{WCB} & \multicolumn{1}{c}{Analytical} & \multicolumn{1}{c}{Cornish-Fisher} & \multicolumn{1}{c}{Actual} \\
\midrule
10  & 2.529 & 2.385 & 4.549 & 2.542 & 2.630 & 3.185 & 3.655 \\
25  & 2.152 & 2.106 & 2.436 & 2.119 & 2.272 & 2.439 & 2.491 \\
50  & 2.051 & 2.030 & 2.184 & 2.028 & 2.139 & 2.195 & 2.227 \\
75  & 2.020 & 2.006 & 2.109 & 2.002 & 2.088 & 2.123 & 2.128 \\
100 & 2.004 & 1.994 & 2.071 & 1.991 & 2.059 & 2.075 & 2.092 \\
200 & 1.982 & 1.977 & 2.014 & 1.972 & 2.013 & 2.020 & 2.037 \\
\bottomrule
\end{tabular}
\endgroup 
\end{table}
 \newpage
 \begin{table}[!htbp]
\centering
\caption{Monte Carlo rejection probabilities in Design 4 (uneven clusters)}
\label{tab:genericfe-en-2-g-10-200-alpha-0p05-nmc-10000-nboot-1000}
 
\begin{tabular}{@{}rrrrrr@{}}
\toprule
\multicolumn{1}{c}{$G$} & \multicolumn{1}{c}{Normal} & \multicolumn{1}{c}{Student $d_1$} & \multicolumn{1}{c}{Pairs} & \multicolumn{1}{c}{WCB} & \multicolumn{1}{c}{Analytical} \\
\midrule
10  & 0.154 & 0.105 & 0.036 & 0.042 & 0.079 \\
25  & 0.085 & 0.069 & 0.044 & 0.056 & 0.052 \\
50  & 0.074 & 0.065 & 0.053 & 0.060 & 0.055 \\
75  & 0.066 & 0.061 & 0.052 & 0.058 & 0.052 \\
100 & 0.059 & 0.056 & 0.050 & 0.055 & 0.050 \\
200 & 0.054 & 0.051 & 0.050 & 0.051 & 0.048 \\
\bottomrule
\end{tabular}
 
\end{table}
 \begin{table}[!htbp]
\centering
\caption{Actual critical values in Design 4 (uneven clusters)}
\label{tab:genericfe-en-2-g-10-200-alpha-0p05-nmc-10000-nboot-1000-critical-values}
\begingroup
 
\setlength{\tabcolsep}{3.5pt}
\begin{tabular}{@{}rrrrrrrrr@{}}
\toprule
\multicolumn{1}{c}{$G$} & \multicolumn{1}{c}{Student $d_1$} & \multicolumn{1}{c}{Student $d_2$} & \multicolumn{1}{c}{Student $d_3$} & \multicolumn{1}{c}{Pairs} & \multicolumn{1}{c}{WCB} & \multicolumn{1}{c}{Analytical} & \multicolumn{1}{c}{Cornish-Fisher} & \multicolumn{1}{c}{Actual} \\
\midrule
10  & 2.385 & 2.385 & 2.778 & 4.085 & 2.932 & 2.655 & 2.983 & 3.218 \\
25  & 2.106 & 2.106 & 2.436 & 2.403 & 2.236 & 2.275 & 2.368 & 2.314 \\
50  & 2.030 & 2.030 & 2.348 & 2.164 & 2.078 & 2.131 & 2.176 & 2.194 \\
75  & 2.006 & 2.006 & 2.319 & 2.096 & 2.034 & 2.080 & 2.096 & 2.096 \\
100 & 1.994 & 1.994 & 2.305 & 2.058 & 2.013 & 2.052 & 2.065 & 2.056 \\
200 & 1.977 & 1.977 & 2.285 & 2.009 & 1.983 & 2.009 & 2.012 & 1.982 \\
\bottomrule
\end{tabular}
\endgroup 
\end{table}
  
\clearpage
\bibliographystyle{econ-econometrica}
\bibliography{main}

\end{document}